\documentclass[%
 aip,
 nofootinbib,
 amsmath,amssymb,
 reprint,%
]{revtex4-1}

\usepackage{bm}
\usepackage[utf8]{inputenc}
\usepackage{amsthm}
\usepackage{graphicx}
\usepackage{amsfonts}
\usepackage{dsfont}
\usepackage{ulem}
\usepackage{float}
\usepackage{makecell}
\usepackage{braket}
\usepackage{fancyvrb}
\usepackage{siunitx}
\usepackage{blkarray}
\usepackage{hyperref}
\usepackage{soul}
\usepackage{xcolor}
\usepackage{graphicx}
\usepackage{dcolumn}
\usepackage{bm}

\usepackage[utf8]{inputenc}
\usepackage[T1]{fontenc}
\usepackage{mathptmx}
\usepackage{etoolbox}

\makeatletter
\def\@email#1#2{%
 \endgroup
 \patchcmd{\titleblock@produce}
  {\frontmatter@RRAPformat}
  {\frontmatter@RRAPformat{\produce@RRAP{*#1\href{mailto:#2}{#2}}}\frontmatter@RRAPformat}
  {}{}
}%
\makeatother

\usepackage{float}
\makeatletter
\let\newfloat\newfloat@ltx

\usepackage{mathtools}

\newcommand{\bs}[1]{\boldsymbol{#1}}
\newtheorem{theorem}{Theorem}

\newtheorem{lemma}[theorem]{Lemma}

\newcommand{\tr}{\mathrm{Tr}}
\newcommand{\n}{\text{nat}}
\newcommand{\ctr}{\text{ctr}}
\newcommand{\Loss}{\mathcal{L}}

\newcommand{\U}{U(\bs{\theta})}
\newcommand{\LE}{\mathcal{L}_E}

\begin{document}
\preprint{AIP/123-QED}

\title{Circumventing Traps in Analog Quantum Machine Learning Algorithms Through Co-Design} 

\author{Rodrigo Araiza Bravo$^\ddagger$}
\email{oaraizabravo@g.harvard.edu}
\affiliation{Department of Physics, Harvard University, Cambridge, MA 02138 USA}

\author{Jorge Garcia Ponce$^\ddagger$}
\email{jorgegarciaponce@college.harvard.edu}
 \affiliation{Department of Physics, Harvard University, Cambridge, MA 02138 USA}
 
\author{Hong-Ye Hu}%
\affiliation{Department of Physics, Harvard University, Cambridge, MA 02138 USA}

\author{Susanne F. Yelin}%
\affiliation{Department of Physics, Harvard University, Cambridge, MA 02138 USA}

\date{\today}

\begin{abstract}
Quantum machine learning QML algorithms promise to deliver near-term, applicable quantum computation on noisy, intermediate-scale systems. While most of these algorithms leverage quantum circuits for generic applications, a recent set of proposals, called analog quantum machine learning (AQML) algorithms, breaks away from circuit-based abstractions and favors leveraging the natural dynamics of quantum systems for computation, promising to be noise-resilient and suited for specific applications such as quantum simulation. Recent AQML studies have called for determining best ansatz selection practices and whether AQML algorithms have trap-free landscapes based on theory from quantum optimal control (QOC). We address this call by systematically studying AQML landscapes on two models: those admitting black-boxed expressivity and those tailored to simulating a specific unitary evolution. Numerically, the first kind exhibits local traps in their landscapes, while the second kind is trap-free. However, both kinds violate QOC theory's key assumptions for guaranteeing trap-free landscapes. We propose a methodology to co-design AQML algorithms for unitary evolution simulation using the ansatz's Magnus expansion. We show favorable convergence in simulating dynamics with applications to metrology and quantum chemistry. We conclude that such co-design is necessary to ensure the applicability of AQML algorithms.
\end{abstract}

\maketitle
\def\thefootnote{$\ddagger$}\footnotetext{These authors contributed equally to this work.}\def\thefootnote{\arabic{footnote}}

\section{Introduction}\label{sec:Introduction}
\begin{figure*}[ht!]
    \centering
    \includegraphics[scale=0.5]{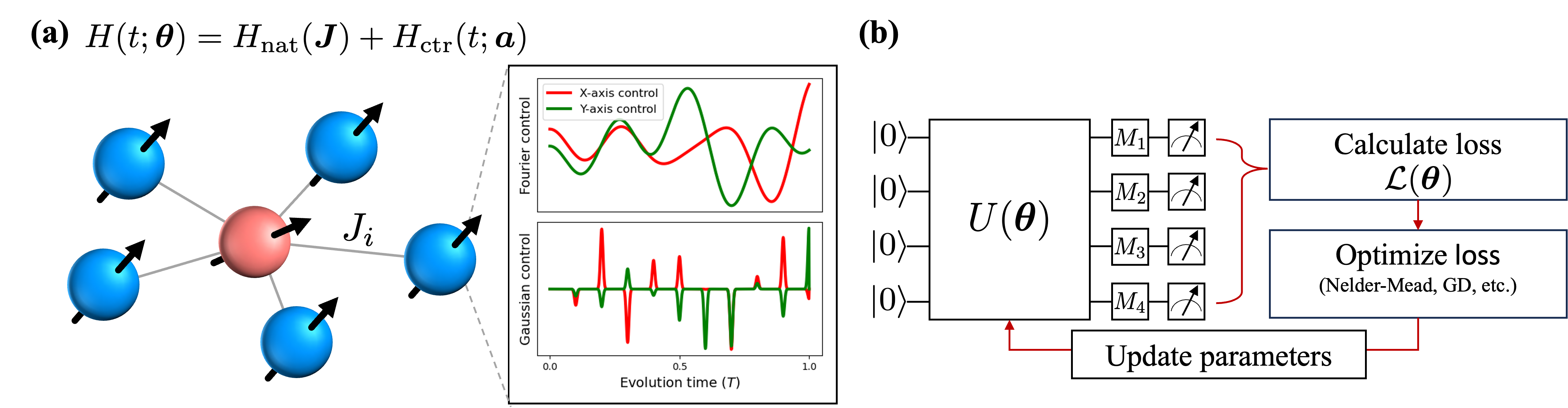}
    \caption{Schematic representation of the AQML Hamiltonian (a) and the learning procedure (b). The analog ansatz's Hamiltonian generally comprises a native Hamiltonian with qubit interactions, parametrized by variables $\bs{J}$,  and a control Hamiltonian of single-qubit terms parametrized by variables $\bs{a}$. The ansatz parameters are $\bs{\theta}=(\bs{J}, \bs{a})$. Panel \textbf{(a)} exemplifies the analog ansatz of a quantum perceptron (QP) composed of $N-1$ qubit interacting with a central qubit. Each qubit experiences several time-dependent control fields sums of simple time-dependent functions such as Fourier and Gaussian functions.
    Panel \textbf{(b)} exemplifies the learning loop with an AQML algorithm. The system undergoes unitary evolution dictated by the time-dependent Hamiltonian with parameters $\bs{\theta}$. The evolved state is then measured to calculate a loss function depending on the variational parameters. A classical optimizer--such as the Nelder-Mead or gradient descent methods--calculates parameters that reduce the loss. The parameters are updated, and the training loop continues until convergence.}
    \label{fig:Schematic}
\end{figure*}

Quantum machine learning (QML) promises to deliver advantageous applications in noisy, intermediate-scale quantum computers. QML algorithms promise to deliver either by using quantum systems to speed up machine learning subroutines~\cite{harrow2009quantum,biamonte2017quantum} or by designing novel techniques to optimize and control noisy quantum systems for machine learning and quantum applications~\cite{cerezo2021variational}. This later approach, called variational quantum algorithms, uses variational optimization for application within near-term devices. These algorithms often comprise an underlying hardware architecture with tunable parameters, a loss function measuring the error relative to a desired computation, and a classical optimizer routine tuning the parameters to minimize the loss. The hardware architecture is often abstracted away and modeled as a digital quantum circuit, and such abstracted algorithms are called variational quantum circuits (VQCs). \par

Despite their promise, VQCs exhibit issues with accuracy, efficiency, and training, which precludes an advantage over classical algorithms. VQCs often experience flat landscapes upon random initialization~\cite{mcclean2018barren}, -- a phenomenon known as \textit{barren plateaus} -- and an exponential number of local minima~\cite{anschuetz2022quantum}. To circumvent these challenges, extensive work has been done on proposing circuit architectures (ans\"atze)~\cite{peruzzo2014variational, farhi2016quantum,beer2020training,yuan2021quantum,martyn2019product}, loss functions \cite{kiani2022learning,larose2019variational,lloyd2020quantum, uvarov2021barren}, regulation techniques~\cite{patti2021entanglement,kobayashi2022overfitting,kuroiwa2021penalty}, and optimization techniques~\cite{chen2022variational, stokes2020quantum, huang2020empirical, bonet2023performance}. \par 

A recent set of proposals, called analog quantum machine learning (AQML) ans\"atze~\cite{tangpanitanon2020expressibility,markovic2020quantum,mujal2021opportunities}, breaks away from circuit-based abstractions. Instead, AQML favors directly using the system's dynamics for computation. At a high level, an AQML ansatz comprises a native interaction Hamiltonian and a set of time-dependent controls. However, AQML studies suffer from various practical drawbacks, such as the fact that simulating time evolution is computationally expensive and, therefore, limits theoretical studies to small system sizes. \par 

AQML appeals to developers for various reasons related to classical analog computation. Firstly, classical analog computing is sometimes advantageous over digital computing for specific applications due to ease of fabrication, problem-solution co-design, and energy cost~\cite{schuman2022opportunities,maclennan2014promise,markovic2020physics}. AQML is thus potentially beneficial for applications related to the simulation of other quantum behavior or for classical problems that may benefit from quantum-like computation. Secondly, analog classical computing can also be noise-resilient upon a judicious choice of learning paradigm~\cite{wunderlich2019demonstrating}. AQML algorithms are potentially resilient to current noise in quantum devices and can thus deliver practical applications sooner. \par 

Recent observations and desires within quantum computation also justify investments in AQML research. Firstly, analog approaches promise to clarify the fundamental computational capabilities of available physical systems~\cite{bravo2022quantum,PhysRevLett.127.100502}. Secondly, small-scale analog systems have been used for various tasks relevant to scientific and industrial communities, including efficient unitary time-evolution simulation~\cite{ghosh2021realising}, sensing~\cite{bravo2022universal}, resource allocation optimization~\cite{leng2022differentiable}, time-series prediction~\cite{mujal2023time}, image classification~\cite{kornjavca2024large}, and memory storage and retrieval~\cite{bravo2022quantum}. Thirdly, while VQCs are expressive, they are also plagued with barren plateaus. On the other hand, for AQML algorithms, the source of barren plateaus stems from excess entanglement~\cite{patti2021entanglement,bravo2022universal, sauvage2024building, park2024hamiltonian,park2024hardware}. Lastly, For both unitary evolution~\cite{kiani2020learning} and ground state finding~\cite{wiersema2020yvette}, shallow VQCs often exhibit low accuracy and exponential local minima~\cite{you2021exponentially}. VQCs more often succeed when the circuit depth is exponential in the system size~\cite{anschuetz2022quantum,liu2023analytic, you2022convergence}. On the other hand, AQML algorithms' landscapes are believed to resemble those from controllable quantum systems more closely. According to quantum optimal control (QOC) theory, these systems exhibit trap-free landscapes~\cite{arenz2014control, arenz2020drawing, wu2012singularities}. Recent AQML studies and perspectives have called for a study of the training landscapes of AQML models~\cite{leng2022differentiable,magann2021pulses}. There is a dire need to determine the topography of the landscapes.  Which problems are suitable for AQML algorithms and what algorithms are ideal for specific problems remain open lines of inquiry.\par 

This work leverages various numerical and analytical studies to interrogate and understand AQML landscapes as a function of the number of control parameters, control types, and system size for a specific loss function. To do this, we focus on simulating the transverse-field Ising model evolution with two groups of ans\"atze: (A1) ans\"atze that have native Hamiltonian different with the Ising interactions but admit black-box universal quantum computation such as the recently introduced quantum perceptrons (QPs)~\cite{bravo2022universal}, and (A2) ans\"atze tailored to the specific tasks such as those with Ising native interactions, or specialized QPs that contain the Ising interaction as a special case (see Sec.~\ref{Sec:Groups} for details). ans\"atze in A1 are expressive, while those in A2 are co-designed with the task we study. Co-design is not a new concept; significant effort has gone into describing how quantum computing architectures are often built to address specific problems \cite{tomesh2021quantum, jiang2021co, zhao2024unraveling}. \par 

In Sec.~\ref{Sec:QOC}, we show that using QOC theory with common QOC assumptions leads to a theoretical prediction of trap-free landscapes when an exponential number of control parameters are available (i.e., in over-parametrized ans\"atze) regardless of the microscopic ansatz details. We also discuss conditions on AQML ans\"atze that lead to a violation of one QOC assumption, \textit{local surjectivity}. Consequently, QOC theory does not apply to AQML ans\"atze. However, numerical experiments show that A2 ans\"atze exhibit trap-free landscapes and better approximation errors while those in A1 are trap-ridden (Sec.~\ref{Sec:Traps}). In summary, we argue that although most AQML landscapes do not possess the amenable qualities of QOC settings, co-designing an ansatz is an important step towards favorable training landscapes.  \par 

In Sec.~\ref{Sec:Codesign}, we show that task-algorithm co-design can be developed for unitary time evolution simulation by analyzing the ansatz's Magnus expansion. We do this analysis pictorially using a generalization of the Lie trees used in optimal quantum control~\cite{arenz2014control}. We show that when co-designed, the QP ansatz is suited for time-reversible spin-squeezing with metrology applications (see Sec.~\ref{Sec:Squeezing}), and the Ising ansatz is suited for unitary couple-cluster evolution with applications to quantum simulation (see Sec.~\ref{Sec:JordanWigner}). In both cases, satisfactory accuracy is reached within a constant number of controllable parameters. Sec.~\ref{Sec:Conclusion} concludes with further research directions. \par 

Upfront, our studies make two significant conclusions: 
\begin{enumerate}
    \item AQML algorithms violate common QOC assumptions. Despite not following QOC theory, co-designed AQML algorithms exhibit trap-free landscapes. On the other hand, black-box algorithms exhibit traps. Further work is needed to determine if co-design is necessary for trap-free landscapes. 
    \item Succeeding at a task is a more likely outcome when AQML algorithms are co-designed. In the case of unitary time evolution, one approach towards co-design is given by analyzing the terms of the AQML algorithm's Magnus expansion. 
\end{enumerate}

\section{Analog Quantum Machine Learning Algorithms}\label{Sec:Algorithm}
This section introduces AQML algorithms comprising an analog ansatz, a loss function, and a classical optimization scheme. We also introduce two classes of algorithms: black-box expressive ones and task co-designed ones. This distinction will allow us to see the quantitative differences in training performance due to co-design.\par 

\subsection{Analog ans\"atze}\label{Sec:Architecture}
Our studies concern a system of $N$ qubits labeled $i=1,...,N$, each living in the Hilbert space $h_i$ spanned by the basis states $\{|0\rangle_i, |1\rangle_i\}$. The system inhabits the product Hilbert space $\otimes_{i=1}^Nh_i$. We call $X_i,Y_i$, and $Z_i$ the spin operators acting on qubit $i$, which obey the usual spin commutator relationship $[X_i, Y_j] = i\delta_{ij}Z_i$~\cite{preskill2015lecture}. The system evolves under a time-dependent Hamiltonian $H(t;\bs{\theta})$ composed of a native--also called drift--($H_{\n}$) and control ($H_{\text{ctr}}$) Hamilonians, with tunable parameters $\bs{J}$ and $\bs{a}$ respectively with $\bs{\theta}=(\bs{J}, \bs{a})$. For the Ising and QP ans\"atze, the native Hamiltonians are
\begin{align}
    H_{\text{I}}(\bs{J}) &= \sum_{i=1}^{N-1} J_i Z_iZ_{i+1},\\
    H_{\text{QP}}(\bs{J}) &= \sum_{i=1}^{N-1} J_i Z_iZ_{N}.
\end{align}
We select the QP Hamiltonian based on our recent introduction and interest in QPs as the basic \textit{motif} of other 2D native Hamiltonians~\cite{bravo2022universal}. QPs are useful for ancillary-based operations, energy measurements, and sensing applications. Indeed, the QML community is interested in proposing scalable and analytically tractable building blocks of quantum neural networks to inform computational capacity~\cite{schuld2022quantum}. Previously, we studied the source and presence of vanishing gradients upon random initialization in QPs. In this study, we are interested in how the parameters $\bs{\theta}$ shape the landscape. \par 
The control Hamiltonian we use is
\begin{align}
   H_{\ctr}(t; \bs{a})&=\sum_{i=1}^N f^x_i(t;\bs{a})X_i+ f^y_i(t;\bs{a})Y_i.
\end{align}
where $\bs{a}$ are variational parameters that modulate the control terms through the functions $f_i^\alpha$. It is customary to parameterize these functions as a finite sum of $K$ simple functions $f_i^\alpha(t) = \sum_{k=0}^{K-1}a_{ik}^{\alpha} g_k(t)$~\cite{leng2022differentiable}. In our case, the functions $g_k$ will be Fourier series, Gaussians, or piece-wise constants. See Appdx.~\ref{A:QPs_IsingEvolution} for details.\par 

We assume that the quantum system is noiseless. Thus, an initial state $|\Psi_0\rangle$ evolves under the Schr\"{o}dinger equation $i\partial_t|\Psi(t)\rangle = H(t;\bs{\theta})|\Psi(t)\rangle$ to the final state $|\Psi_f\rangle$ at $t=T$. Mathematically, the final state is a unitary transformation of $|\Psi_0\rangle$ defined by
\begin{align}
    U(\bs{\theta}) &= \mathcal{T}\exp\left(-i\int_{0}^T H(\tau;\bs{\theta})d\tau\right) \label{eq:FormalUnitary}
\end{align}
where $\mathcal{T}$ is the time-ordering operator. To recap, an AQML ansatz is a particular choice of native Hamiltonian and controls.  Fig.~\ref{fig:Schematic} shows a pictorial representation of the QP ansatz.  \par 

\subsection{Learning with Analog ans\"atze}
In the most general setting, after the system evolves under $\U$, the system is measured several times with each qubit in a particular basis $M_i$ to calculate a loss $\Loss$ encoding the error relative to a target behavior. The loss is then passed through a classical optimizer for improvement through tuning the parameters $\bs{\theta}$. \par 

We focus on learning a unitary time evolution $U^{\text{targ}}$. The motivation for this loss function is two-fold. Firstly, studying this loss tells us about the possibilities of using AQML algorithms for quantum simulation, a promising application of quantum hardware. Secondly, previous QOC studies have shown that this loss has a trap-free landscape for some simple Hamiltonians. Mathematically, the loss is the Fr\"obenius norm
\begin{align}\label{eq:LossE}
    \mathcal{L}_{E}(\bs{\theta}) &= \frac{1}{2^{N+1}}||\U-U^{\text{targ}}||^2,
\end{align}
where $||A||_F^2 = \mathrm{Tr}(A^\dagger A)$, and the factor of $2^{-N-1}$ bounds the loss to the interval $[0,1]$ making $\mathcal{L}_E$ a measure of fidelity. \par 

Except for Sec.~\ref{Sec:Codesign}, the target unitary in all our numerical experiments is generated by the 1D transverse field Ising Hamiltonian 
\begin{equation}\label{eq:IsingEvolution}
    U^{\text{targ}} = \exp\left(-i\left(\sum_{i=1}^{N-1} Z_iZ_{i+1}-0.1\sum_{i=1}^N X_i\right)\right).
\end{equation}
The parameters will be updated using Adam, a version of gradient-descent with momentum and scalability terms~\cite{kingma2014adam}. 

\subsection{Groups of Anzatzes}\label{Sec:Groups}
We now define the groups of ans\"atze we alluded to in the introduction. All ansatz we study admit admits full expressivity in the form of controllability. These are ans\"atze for which the operators in the native and control Hamiltonians can produce any element of the dynamical Lie algebra (i.e., Pauli strings in our case) through nested commutations. \par 

Group A1 comprises QPs with arbitrary control parameters and only Gaussian or Fourier controls. These ansatz are expressive but not tailored for the task in Eq.~\ref{eq:IsingEvolution}. Group A2 consists of expressive ans\"atze with a close form solution for simulating the evolution in Eq.~\ref{eq:IsingEvolution}. For example, the Ising ansatz with constant fields is immediately in A2. Moreover, in Appdx.~\ref{A:QPs_IsingEvolution} we show that specialized QPs with a combination of Gaussian controls and extra piece-wise constant controls on the output qubit are in A2. This can also be derived from analyzing the ansatz Magnus expansion as in Sec.~\ref{Sec:Codesign}.
In summary, algorithms in A1 and A2 are expressive, but those in A2 are tailored to the task in Eq.~\ref{eq:IsingEvolution}.\par 

In Sec.~\ref{Sec:Landscapes} we show that QPs in A1 perform purely at the task in Eq.~\ref{eq:IsingEvolution} while specialized A2 QPs perform better but worse than the Ising ansatz. However, we show that A2 QPs, co-designed for the task, exhibit landscapes free of local minima. Thus, we contrast A1 and A2 ans\"atze to showcase the importance of co-design. 

\section{Landscapes of Analog Quantum Machine Learning Algorithms}\label{Sec:Landscapes}
Before going any further, this section introduces crucial concepts in landscape analysis. We discuss theoretical conditions from QOC that are necessary for trap-free landscapes, and show that while these conditions are violated, co-designed ans"atze in A2 indeed show trap-free landscapes. \par 

Let us begin reviewing how to characterize a landscape's curvature. In general, studying the landscape of any algorithm encompasses finding the parameter values with vanishing loss's derivative (i.e., the critical points). Critical points are then categorized as minima, maxima, or saddles, depending on the curvature around them. \par 

The derivative can be expressed as the function composition $\partial_{\bs{\theta}}\mathcal{L} = \partial_U \mathcal{L}\circ \partial_{\bs\theta}U$. For the case of Eq.~(\ref{eq:LossE}), the first term is independent of $\bs{\theta}$, and therefore all critical points satisfy $\mathrm{Tr}(\partial_{\bs\theta}U)=0.$\par 

In quantum optimal control theory (QOC), $\partial_{\bs\theta}U$ is called the \textit{dynamical derivative}, and it captures how the unitary $U$ changes for a changing control field. For an AQML ansatz, the dynamical derivative of a parameter $\theta^{\alpha}_{ik}$ associated with the time-dependent control function $g_k(t)$, qubit $i$ and the Hamiltonian operator $h_i^\alpha$, is given by (see Appdx.~\ref{A:Dynamical} for details)

\begin{equation}\label{eq:dynamical_derivative}
    \frac{\partial U}{\partial \theta_{ik}^\alpha} = -i U(\bs{\theta})\int_{0}^T d\tau g_{k}(\tau)h_i^\alpha(\tau).
\end{equation} 
Here, if the parameter is associated with a term in $H_{\text{nat}}$, we use $g_k(t)=1$ and $h_i=Z_iZ_{k}$; else $h_i^\alpha = X_i$ or $Y_i$. Eq.~\ref{eq:dynamical_derivative} is a special case of the AQML dynamical derivative in Ref.~\cite{leng2022differentiable} assuming our particular ansatz architecture.\par 

The nature of the critical points is dictated by the Hessian (i.e., second derivative) of the loss. In the case of the case of Eq.~(\ref{eq:LossE}), the Hessian is given by $\partial_{\bs{\theta'}}\partial_{\bs{\theta}}\mathcal{L} = \partial_{U}\mathcal{L}\circ \partial_{\bs{\theta'}}\partial_{\bs{\theta}} U$. The eigenvalues of the Hessian determine the nature of the critical point (see Fig.~\ref{fig:approximations}(a) for examples). A trap-free landscape has only one maximum or minimum critical point, while the rest are saddles. See Fig.~\ref{fig:approximations}(a) for an example of a trap-free landscape. \par  

For the remainder of this section, we show that AQML algorithms violate one of QOC's conditions necessary to ensure trap-free landscapes (Sec.~\ref{Sec:QOC}). However, we present evidence that co-design ans\"atze in A2 exhibit trap-free landscapes (Sec.~\ref{Sec:Traps}). 

\subsection{AQML Theory of Trap-Free Landscapes}\label{Sec:QOC}

\begin{figure}[b]
\includegraphics[width=0.95\linewidth]{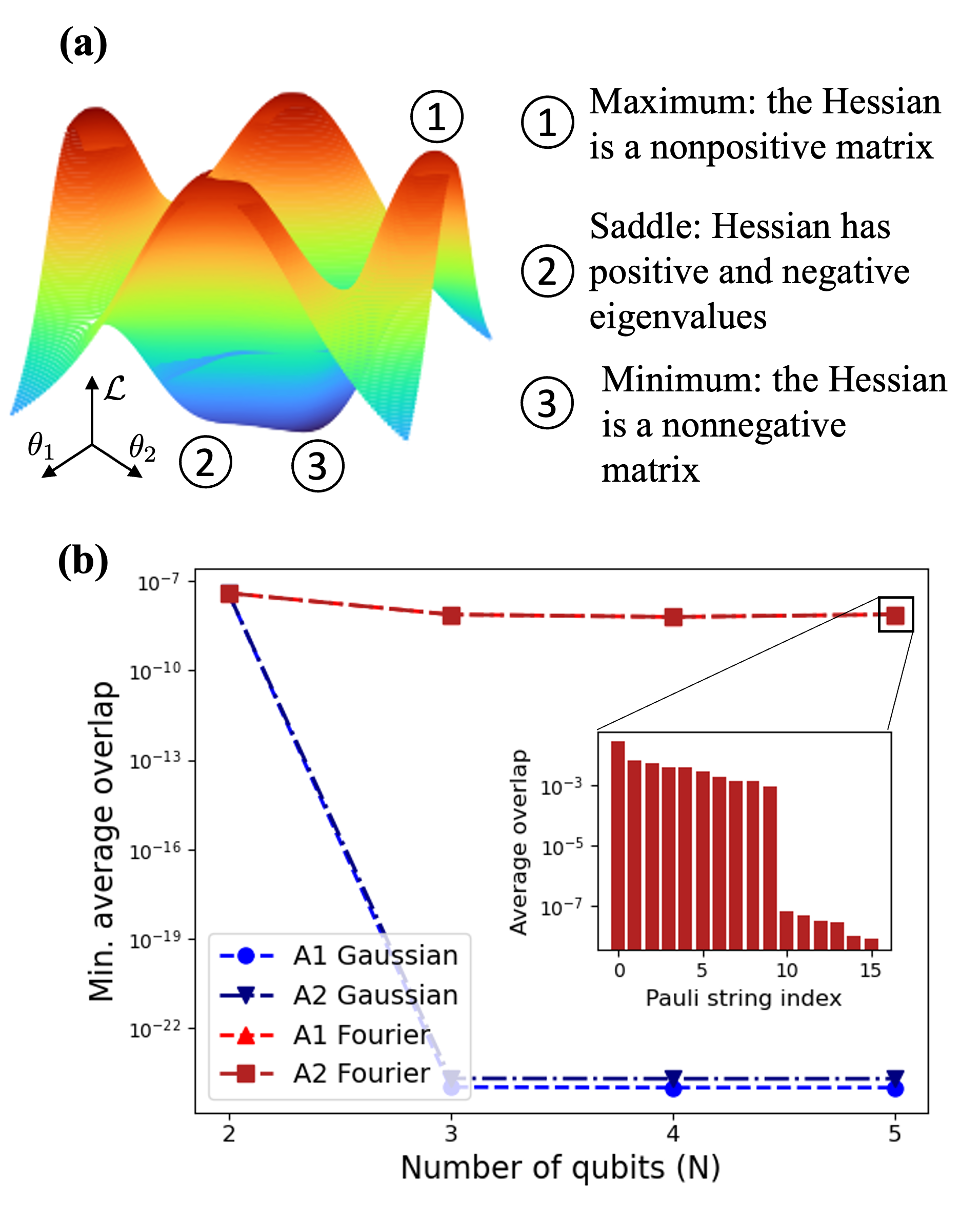}
\caption{The optimization landscape depends on the ansatz and the loss. \textbf{(a)} Exemplifies a trap-free optimization landscape where all sub-optimal critical points are maximums and saddles, while the only minimum is optimal. According to QOC theory, the AQML ansatz we present could be trap-free if it satisfies the condition of local surjectivity, meaning that each dynamical derivative of the unitary has a nonzero average overlap with each Pauli string. \textbf{(b)} Shows numerical calculations of the average overlaps between the dynamical derivatives and Pauli strings for A1 and A2 QP ans\"atze with $N=2,3,4,5$ qubits and with $4^N-1$ parameters. For some Pauli strings, the average overlap is minimal at zero or near zero, violating local surjectivity.}
\label{fig:approximations}
\end{figure}
We now discuss conditions from QOC theory necessary for trap-free landscapes. We show that AQML ans\"atze in both A1 and A2 violate these requirements, but we will show that A2 ans\"atze are trap-free. \par 

According to QOC theory, three conditions are sufficient to ensure a trap-free landscape~\cite{ho2009landscape,rabitz2006optimal,russell2017control}, which several studies have assumed to apply to AQML algorithms~\cite{magann2021pulses,leng2022differentiable}. The conditions are:
\begin{enumerate}
    \item \textbf{Unconstrained Fields.} The control fields $f_n^\alpha$ should be allowed to take on any real value.
    \item \textbf{Controllability.} The terms in the $H(t)$ should produce any element of the Hilbert space's dynamical Lie algebra through nested commutators. 
    \item \textbf{Local Surjectivity.} The dynamic derivatives of the ansatz for every parameter must be full-rank in the dynamical Lie algebra. 
\end{enumerate}
For AQML, the dynamical Lie algebra comprises all possible Pauli strings $P$ composed of products of $X_n, Y_n, Z_n$, and the identity operator. Local surjectivity requires that every Pauli-string $P$ has a nonzero average overlap with the dynamical derivative in Eq.~(\ref{eq:dynamical_derivative}). That is, 
\begin{equation}
    \forall P \quad \mathbb{E}_{\bs{\theta}}\mathrm{Tr}(P\partial_{\bs{\theta}}U)\neq 0.
\end{equation}\par 

The first two conditions are satisfied for our AQML ans\"atze. Indeed, the control fields are unconstrained, and the ans\"atze allow for universal quantum computation~\cite{gao2017quantum,bravo2022universal}, which is linked to controllability \cite{ramakrishna1996relation}. On the other hand, local surjectivity must be checked numerically on all $4^N-1$ Pauli strings. As a result, local surjectivity is often assumed. Theoretical studies of QOC landscapes justify this assumption because, experimentally, optimal control fields are indeed easy to find in most cases~\cite{ho2006effective}. However, this easiness is not the case for AQML algorithms.\par 

In Appdx.~\ref{A:LE_Landscapes}, we show that if one assumes these three conditions for an AQML ansatz with at least $4^N-1$ variational parameters, then the landscape of $\mathcal{L}_E$ is trap-free (i.e., it consists of saddles and only one local minimum and maximum). This proof uses the methods developed in Ref.~\cite{ho2009landscape}. Using the techniques in Ref.~\cite{rabitz2006optimal}, one could derive a similar result for the loss related to ground state preparation. \par 

Fig.~\ref{fig:approximations}(b) shows numerical tests of local surjectivity using A1 and A2 QP ansatz with Fourier and Gaussian parametrizations for different qubit numbers. We see that while the average overlap for many Pauli's is nonzero, there exist Pauli strings for which the overlap is at a minimum of zero or close to zero. In other words, The ans\"atze fail the local surjectivity condition even in the overparametrized regime of $4^N-1$ parameters. We also observe similar results for fewer parameters and other parametrizations, such as using Legendre polynomials and piece-wise constant functions.\par 

This observed violation gives way to a natural question: under what conditions can we expect an AQML algorithm to violate local surjectivity? In Appdx.~\ref{A:LSurj}), we show that when an AQML ansatz produces a unitary $U(\bs{\theta})$ that is Haar random distributed, then local surjectivity is violated for every $P$. This condition can be relaxed to $U(t;\bs{\theta})$ following a unitary 1-design in case the control functions $g_k$ vanish around $t=0, T$, in which cases local surjectivity is violated for Pauli's $P$ contained in the native Hamiltonian. We highlight that these conditions are sufficient but not necessary. Our algorithms do not resemble Haar randoms. Thus, further work is needed to close the gap between our analytic and numeric understanding of the scenarios leading to the violation of local surjectivity. However, as the next section shows, trap-free landscapes can still appear when an ansatz is co-designed, even when local surjectivity is violated. \par 

\subsection{Numerical Evidence for Trap-Ridden and Trap-Free Landscapes}\label{Sec:Traps}
We now present a numerical analysis of the landscapes of A1 and A2 ans\"atze. We show that A2 QPs are trap-free. We compare the quality of the solutions produced by both A1 and A2 QPs at approximating the time evolution in Eq.~\ref{eq:IsingEvolution}. Importantly, this evolution is theoretically simulatable since both an\"ate are expressible, but only A2 QPs are co-designed to accomplish this task. The main results are presented in
Fig.~\ref{fig:a1_a2_averages} and Fig.~\ref{fig:a1_a2_hessians}. \par 

\begin{figure*}[ht!]
\centering
\includegraphics[width=\linewidth]{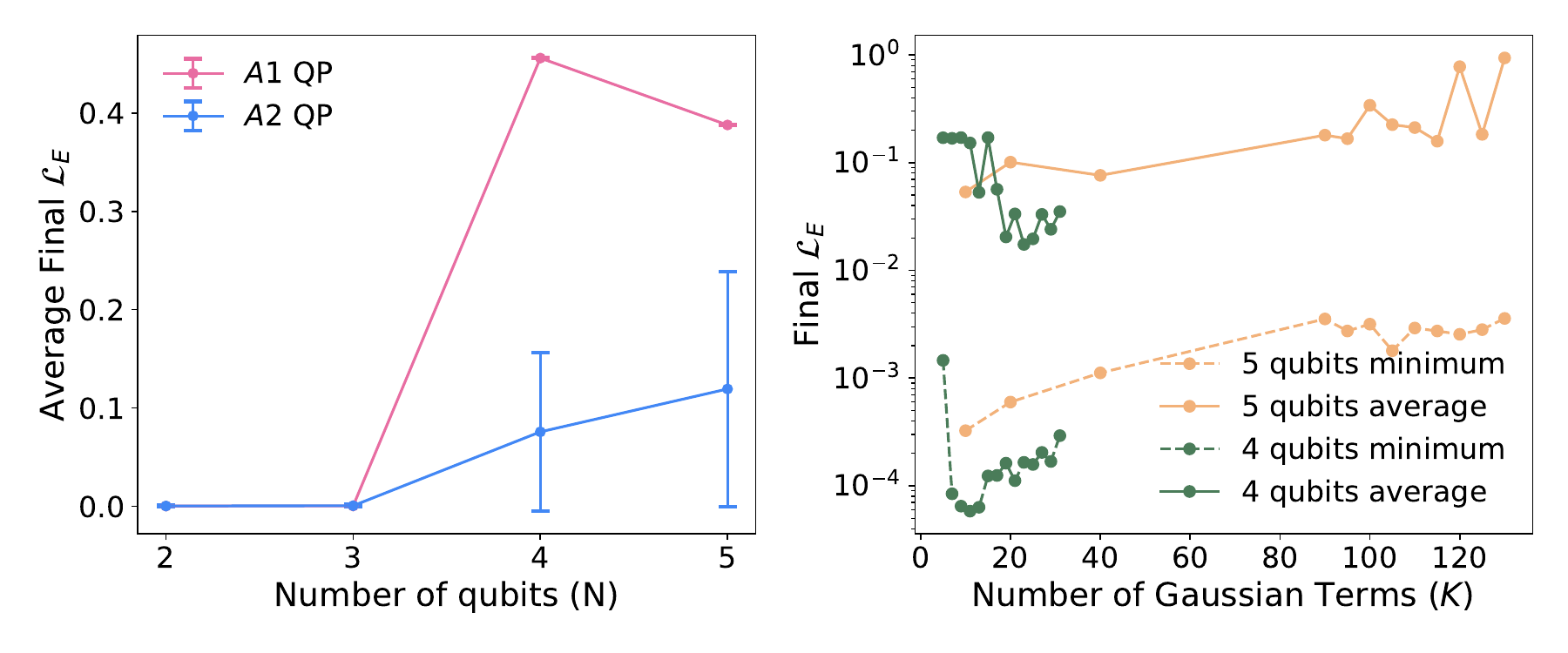}
\caption{Approximating the 1D transverse field Ising evolution on type A1 and A2 QPs. Panel (a) shows the average Fr\"obenius norm found for both types of QPs after $100$ trials and for different qubit numbers. We observe that both types of QPs fail to approximate the desired time evolution. A2 QPs consistently outperform A1 QPs. To further investigate the A2 QPs behavior, (b) shows the average and minimum Fr\"obenius norm for A2 QPs with increasing control functions $K$ for $N=4$ and $5$. For $N=4$ qubits, we observe an improvement in the loss as we approach the over-parametrized regime ($K=26$), followed by an increasing error due to over-fitting. For $N=5$ qubits, this improvement is no longer present. Moreover, when comparing the average loss to the minimum loss, we observe a stark difference by approximately two orders of magnitude, thus suggesting a wide spread of the loss at the critical points.}
\label{fig:a1_a2_averages}
\end{figure*}

For Fig.~\ref{fig:a1_a2_averages}(a), we initialized 100 instances of each QP ansatz with $K =5N$ and $K=3N$ control fields for A1 and A2 ans\"atze, respectively. We note that we ran several more simulations with varying numbers of fields, and the simulation
results remained practically unchanged. Fig.~\ref{fig:a1_a2_averages}(a) shows the average converged loss. We observe that at $N=4$ onwards, both ans\"atze failed to converge, on average, to the optimal solution. We attribute the sudden change at $N=4$ to the fact that both A1 and A2 ans\"atze have a ``star-like'' architecture, which reduces to the Ising connectivity at the values of $N=2$ and $N=3$. At $N\geq 4$, however, neither can reliably converge to the optimal solution. We note that the A2 ansatz consistently outperforms the A1, and the error bars show a much larger spread in the loss distribution at the converged critical points. \par 

To further investigate the effect of the number of control functions ($K$) and obtain a better picture of the quality and spread of the solutions for these landscapes, we repeated the same experiments for the A2 QP ansatz with varying $K$.  Fig.~\ref{fig:a1_a2_averages}(b) shows the resulting plots of minimum and average over $100$ trials for each value of $K$. For $N=4$, we see an improvement in the average and minimum loss found as $K$ approaches the over-parametrized regime at $K=26$, followed by an increasing error due to over-fitting. In addition, we observe a stark difference between the average and minimum solutions found by approximately $2$ orders of magnitude. This contrasts with the homogeneity in loss values observed for the A1 QPs. For $N=5$, the sudden improvement as the pulse basis sets grow in size is absent. Yet, we still recover the trend in the spread of the found critical points, signaling a richer and more heterogeneous distribution of critical points with differences in quality by approximately $2$ orders of magnitude. The explored values of $K$ encompass the under- and over-parametrized regimes for each value of $N$ ($K=26$ and 85, respectively). 

\begin{figure*}[ht!]
\centering
\includegraphics[width=\linewidth]{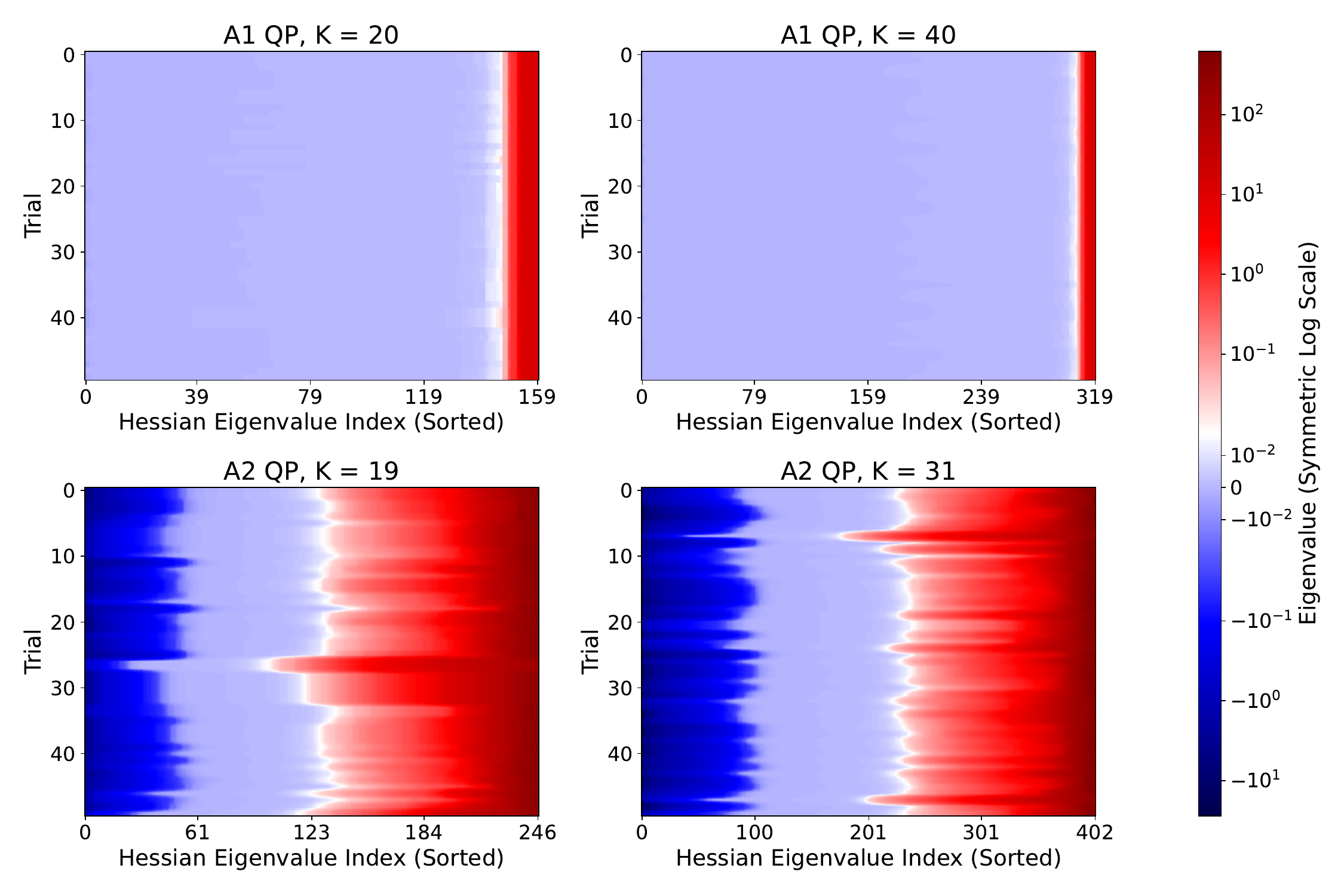}
\caption{Trap-free AQML landscapes for A2 ans\"atze. Each plot shows the magnitude of the sorted Hessian eigenvalues for $50$ trials calculated through automatic differentiation for $N=4$ of A1 and A2 QPs. The basis set sizes ($K$) are chosen such that the ans\"atze on the left (right) column are under- (over-) parametrized. All eigenvalues are non-negative for the first row (A1 QP ansatz); characterizing these critical points as traps. This is no longer the case for the second row (A2 QP ansatz), where we observe the appearance of positive and negative eigenvalues. The critical points found through A2 present positive and negative curvature along some directions, thus characterizing them as saddle points despite this ansatz violating key assumptions in QOC theory.}
\label{fig:a1_a2_hessians}
\end{figure*}

We investigate the nature of the critical points we found for both A1 and A2 QPs. Using automatic differentiation, we calculate the exact Hessian matrix at each critical point across 50 trials with varying numbers of control parameters $K$. Fig.~\ref{fig:a1_a2_hessians} shows the results for four select choices of $K$ corresponding to below and above the over-parametrized regime of $4^N-1$ parameters (see Appdx~\ref{A:Traps} for more details). The right (left) column plots the Hessian eigenvalues for both ans\"atze in the under- (over-) parameterized regime. As we can see, for the A1 QPs, both regimes exhibit only non-negative eigenvalues; thus, these critical points correspond to local minima. For the A2 QPs, however, we observe the presence of both positive and negative eigenvalues in both regimes. The critical points of A2 QPs can thus be classified as saddles. This result adds to the difference in solution quality achievable through both ans\"atze classes, as it further characterizes the A2 QPs solutions as saddles that can potentially be avoided by increasing the number of optimization cycles.\par 

In sum, our numerical experiments provide evidence of the absence of amenable landscapes for black-boxed expressive algorithms. In particular, the A1 QPs are swamped with traps. However, the experiments also revealed a quantitative and qualitative difference in the achievable solutions through the different ans\"atze. Specifically, while neither class was able to produce an optimal solution on average, we observed a significant increase in the variability of the quality of the solutions achievable through the A2 QPs and, for both cases, a significantly better solution (by approximately $3$ orders of magnitude when compared to the A1 QP solutions) was found. Additionally, we found that the A2 QPs produced saddle points, which opens the possibility of improving the solutions found by optimizing them for longer. These drastically different results in solution quality thus hint at the importance of choosing the right ansatz for the appropriate problem.

\section{Task-Algorithm Co-design}\label{Sec:Codesign}
We have shown that AQML algorithms suffer from trap-ridden landscapes. Thus, an outstanding set of questions remains: What kinds of unitaries can an AQML algorithm readily approximate? Conversely, given a desired unitary, can we devise a systemic methodology for allocating attention to different ansatz? This section argues that these questions can be addressed by thinking of algorithm-task co-design. \par 
Co-design is not a new concept; significant effort has gone into describing how quantum computing architectures are often built to address specific problems \cite{tomesh2021quantum, jiang2021co, zhao2024unraveling}. More generally, technologies are defined by, and help define, the issues they aim to resolve~\cite{pinch2012social}. Co-design is crucial in the development of hardware control software \cite{li2021co}, in proposals for materials and chemistry simulation \cite{RevModPhys.86.153,maskara2023programmable,arguello2019analogue,kivlichan2018quantum}, and in approaches to error correction \cite{bluvstein2024logical,PhysRevLett.129.030501, campbell2024series}. Indeed, the field of AQML broadly asserts that the applications developers envision are inseparable from the hardware we expect to use to realize them. \par 
In this section, co-design will take the following particular meaning. An AQML ansatz is co-designed for a desired unitary evolution such that its Magnus expansion contains the operators of the desired evolution weighted by independently tunable coefficients. Likewise, understanding the Magnus expansion of an AQML ansatz can co-design a desired unitary with a specific application. 

\subsection{Co-design guided by the Magnus Expansion}
\begin{figure*}[ht!]
    \centering
    \includegraphics[width=0.95\linewidth]{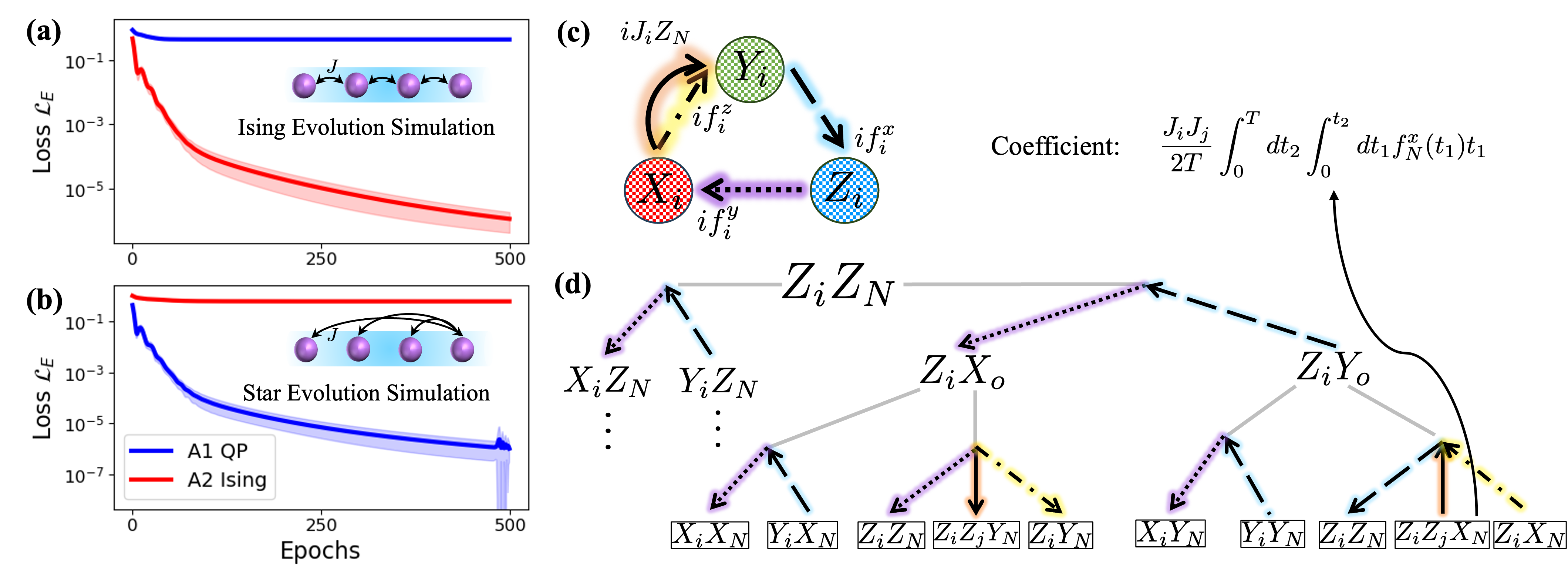}
    \caption{The importance of task-algorithm co-design. In (a)-(b), we train an A1 QP and an A2 Ising ansatz at two different tasks: simulating the evolution under a 1D transverse-field Ising evolution and of a transverse-field star evolution. Panel (a) The Ising evolution results demonstrate that the Ising ansatz succeeds while the QP ansatz trails behind. Panel (b) Shows that the QP is better suited for simulating the star evolution than the Ising ansatz. Using an ansatz suited for a particular task is key to success. Panels (c)-(d) pictorially depict calculating the operators and their coefficients in the Magnus expansion to co-design ans\"atze for a particular task.  Panel (c) depicts the transformation of different Hamiltonian terms in the QP ansatz due to nested commutations. Single-site operators transform to other single-site operators when commuted with control fields. The interaction terms in $H_{\text{nat}}$ generate many-body terms. Going against the arrows results in the coefficients picking up an extra factor of $-1$. Panel (d) exemplifies the operators in the Magnus expansion emerging from commuting an interaction term twice. The resulting operators are all in the $l=2$ term of the Magnus expansion. The arrows here symbolize the coefficients associated with each operator. We exemplify the coefficient in front of $Z_iZ_jX_N$, which can be used for time-reversible spin-squeezing.}
    \label{fig:codesign}
\end{figure*}
Let us begin by justifying our insisted attention to the Magnus expansion, and showing how it can be used for co-design.\par 
A time-dependent Hamiltonian produces a unitary $U(\bs{\theta})$ which can be expressed as generated by an effective Hamiltonian $U = \exp\left(-iT H_{\text{eff}}(\bs{\theta})\right).$ This effective Hamiltonian is calculated using the Magnus expansion \cite{brinkmann2016introduction, choi2020robust}: 
\begin{equation}
    H_{\text{eff}}(\bs{\theta}) = H^{(0)}(\bs{\theta})+H^{(1)}(\bs{\theta})+H^{(2)}(\bs{\theta})+\hdots
\end{equation}
each term in the expansion comprises nested commutators of the control and native Hamiltonian of the AQML algorithm (see Appdx.~\ref{A:Codesign} for details). The term $H^{(0)}$ is called the average Hamiltonian: 
\begin{align}
    H^{(0)}(\bs{\theta}) &= \frac{1}{T}\int_0^Tdt H(t;\bs{\theta})\notag \\
    &= H_{\text{nat}} +\sum_{i}F_i^x(\bs{\theta}) X_i+F_i^y(\bs{\theta}) Y_i+F_i^z(\bs{\theta}) Z_i.
\end{align}
where $F_i^\alpha$ is the integral of $f_i^\alpha$ divided by $T$. \par 

Unsurprisingly, we can expect an ansatz to be reliably used to simulate the dynamics of a many-body Hamiltonian whose interactions resemble those in $H_{\text{nat}}$. Fig.~\ref{fig:codesign} shows the A1 QP and A2 Ising ansatz results in approximating two kinds of evolution. Fig.~\ref{fig:codesign}(a) shows the training results of approximating a 1D transverse-field Ising model as in Eq.~\ref{eq:IsingEvolution} with $h=0.5$. The QP ansatz performs poorly, while the Ising ansatz is successful. Likewise, Fig.~\ref{fig:codesign}(b) shows the training results of approximating a star model with a transverse field of strength $h=0.5$. In the star evolution, all qubits interact with a central qubit. This kind of evolution was, in fact, the inspiration for the QP connectivity. The Ising ansatz performs poorly, while the QP ansatz is successful.\par 

For both experiments, $N=4$ and 100 randomly initialized were optimized to approximate a $T=1$ evolution with $K=N$. We observe that convergence is achieved when $a_{ik}^\alpha \approx 0$ for all $k>0$ (i.e., the algorithm is trained to choose constant fields) and for $\alpha=y,z$. Importantly, for the case of constant fields, $H^{(l)}=0$ for all $l>0.$ \par

The example of co-design in Fig.~\ref{fig:codesign} is very simple. However, the idea is that by analyzing the terms in $H^{(l)}(\bs{\theta})$ for $l>0$, and the coefficients in front of them, we can tell what kinds of unitaries can be readily approximated by our AQML algorithm. Notice that this is quite different than ensuring simulatability through either a claim of universality or controllability. In the case of universality, a desired unitary $U^\text{targ}$ is broken down into a product of unitaries, each achievable given a particular hardware architecture. Instead, our approach focuses on determining if the Hamiltonian generating the desired unitary is spaned by the term within the ansatz's Magnus expansion. Controllability studies whether the nested commutators of the terms in the Hamiltonian can generate every generator of the dynamical Lie algebra (Pauli strings, in our case). Our approach also attends to the coefficient before a given generator, which influences how easy it is to achieve evolution under said generator using an AQML ansatz in an experiment. By looking at the Magnus expansion, one can gain insight into both the generators and the coefficients in front of the generators, which are crucial to determining the likelihood of success in practice. In summary, if $A$ is an operator present in the Magnus expansion with a coefficient $c(\bs{\theta})$, we can, in theory, approximate the evolution $\exp\left(-ic(\bs{\theta})A\right)$ as long as $c(\bs{\theta})$ can be made nonzero all other coefficients for other operators can be mitigated. \par 

Fig.~\ref{fig:codesign} exemplifies how to calculate the operators and coefficients in the Magnus expansion for $l=2$ for a QP. Fig.~\ref{fig:codesign}(c) is a pictorial representation of how the Hamiltonian terms transform into each other through commutation. For example, the operator $Z_i$ transforms into $Y_i$ while the associated coefficient picks a $if_i^z(t)$ factor. The yellow, short-dash, long-dash arrow in Fig.~\ref{fig:codesign}(c) depicts this example. Going in the direction opposite to the arrows picks up an extra $-1$ on the coefficient. Fig.~\ref{fig:codesign}(d) exemplifies how to use this pictorial depiction of the commutation-induced transformations to compute the operators generated from the interaction $Z_iZ_N$ appearing in $H^{(2)}.$ We note that new interactions emerge between two inputs and the output qubit (i.e., three body terms). Using panel (c), each term comes with a coefficient composed of a nested integral of control fields and interaction strengths (see Appdx.~\ref{A:Codesign} for details on how to calculate the coefficients). \par 

It is important to note that the order of integration matters for these coefficients. For an example of the importance of integration order, see Eq.~(\ref{eq:Example_Integration}) in the Appdx.~\ref{A:Codesign}. Using this method and focusing on a particular term in the Magnus expansion label by $l$ gives us a function mapping the variational parameters $\bs{\theta}$ to the coefficients of all possibly generated Hamiltonian terms. Fig.~\ref{fig:Commutators_Appdx} shows all operators obtained from $Z_iZ_N$ in the $l=2$ term of the expansion. \par

Fig.~\ref{fig:codesign}(c) also exemplifies the coefficient in front of the operator $Z_iZ_j X_N$, creating an effective interaction between two input qubits mediated through the output qubit. In Sec.~\ref{Sec:Squeezing}, we show that this operator can naturally be used for time-reversible spin-squeezing, an observation that was previously made in Ref.~\cite{bravo2022universal} using second-order perturbation theory.\par 

Each term in the Magnus expansion contains $\mathcal{O}(N^{2l})$ operators. At first glance, for simple control fields (i.e., small $K$), one would expect that the coefficients of different operators are linearly dependent. However, numeric experimentation shows that the operators' coefficients are linearly independent and can thus be, in theory, tuned independently. For this, see Fig.~\ref{fig:SVDs} in the Appdx.~\ref{A:Codesign}, where we also argue that integration order relevance ensures linear independence by analyzing operators generated up to $l=2$. \par

With this framework in mind, let us show that co-design via attention to the Magnus expansion can reveal naturally suitable evolution for a given ansatz. Below, we show that QPs are naturally suited to implement time-reversible spin-squeezing and that the Ising ansatz is naturally well-suited for evolution under products of Wigner-Jordan strings, an evolution paramount in quantum chemistry applications of quantum computers. 

\subsection{Time-Reversible Spin-Squeezing Through Co-design}\label{Sec:Squeezing}
This subsection offers a fresh perspective on an observation we made in previous work (Ref.~\cite{bravo2022universal}), namely that QPs are suited for realizing time-reversible spin-squeezing, which can be used for quantum metrology applications. Our previous observation used second-order perturbation theory. This section shows that the Magnus expansion tells us QPs can do time-reversible spin-squeezing. We do not show how this can be used for metrology applications since that has already been explored. Instead, this section explains how the Magnus expansion can lead us to similar conclusions without the need for perturbation theory. \par 

Fig.~\ref{fig:codesign}(c) shows that the coefficient of the operator $Z_iZ_jX_N$ is given by 
\begin{align}
    \alpha_{ZZX} &= \frac{J^2}{2T}\int_0^Tdt_2\int_{0}^{t_2}dt_1f_N^x(t_1)t_1\notag \\
    &= \frac{J^2}{2T}\left(\Delta G_N^x- \Delta H_N^x\right)
\end{align}
where $G^x_N(t)$ ($H^x_N(t)$) is the second (third) anti-derivative of the function $f_N^x$, and $\Delta G_N^x$ ($\Delta H_N^x$) is the difference of the second (third) anti-derivatives from $t=0$ to $T$. \\
Then, a term in the Magnus expansion is (Appdx.~\ref{A:Squeezing})
\begin{equation}\label{eq:Squeezing}
    H_{\text{eff}} 
     = (F_N^x+\alpha_{ZZX})\left(S_{\text{in}}^z\right)^2X_N.
\end{equation}
Here, $F_N^x$ is the anti-derivative (divided by $T$) of $f_N^x$. In the equation above, $S_{\text{in}}^z = \sum_{i=1}^{N-1}Z_i$ is the total spin z-projection of the input qubits. The term $\left(S_{\text{in}}^z\right)^2$ is called the one-axis-twisting~\cite{PhysRevA.47.5138}, which creates spin-squeezing with applications to Heisenbert-limit metrology~\cite{block2023universal}. Moreover, the inclusion of $X_N$ means that the output qubit can be used to squeeze the input when initialized in the state $|+\rangle_N = \frac{1}{\sqrt{2}}(|0\rangle_N+|1\rangle_N)$, the positive eigenstate of $X_N$. Similarly, the output can anti-squeeze (i.e., create a time-reversed evolution) when in the state $|-\rangle_N = \frac{1}{\sqrt{2}}(|0\rangle_N-|1\rangle_N)$, the negative eigenstate of $X_N$. This time reversal has proven an efficient methodology in metrology experiments~\cite{colombo2022time}. \par 
Thus, the Magnus expansion analysis can elucidate the applications suitable for AQML ansatz. 

\subsection{Products of Jordan-Wigner Strings Through Co-design}\label{Sec:JordanWigner}
\begin{figure}[t]
    \centering
    \includegraphics[width=0.90\linewidth]{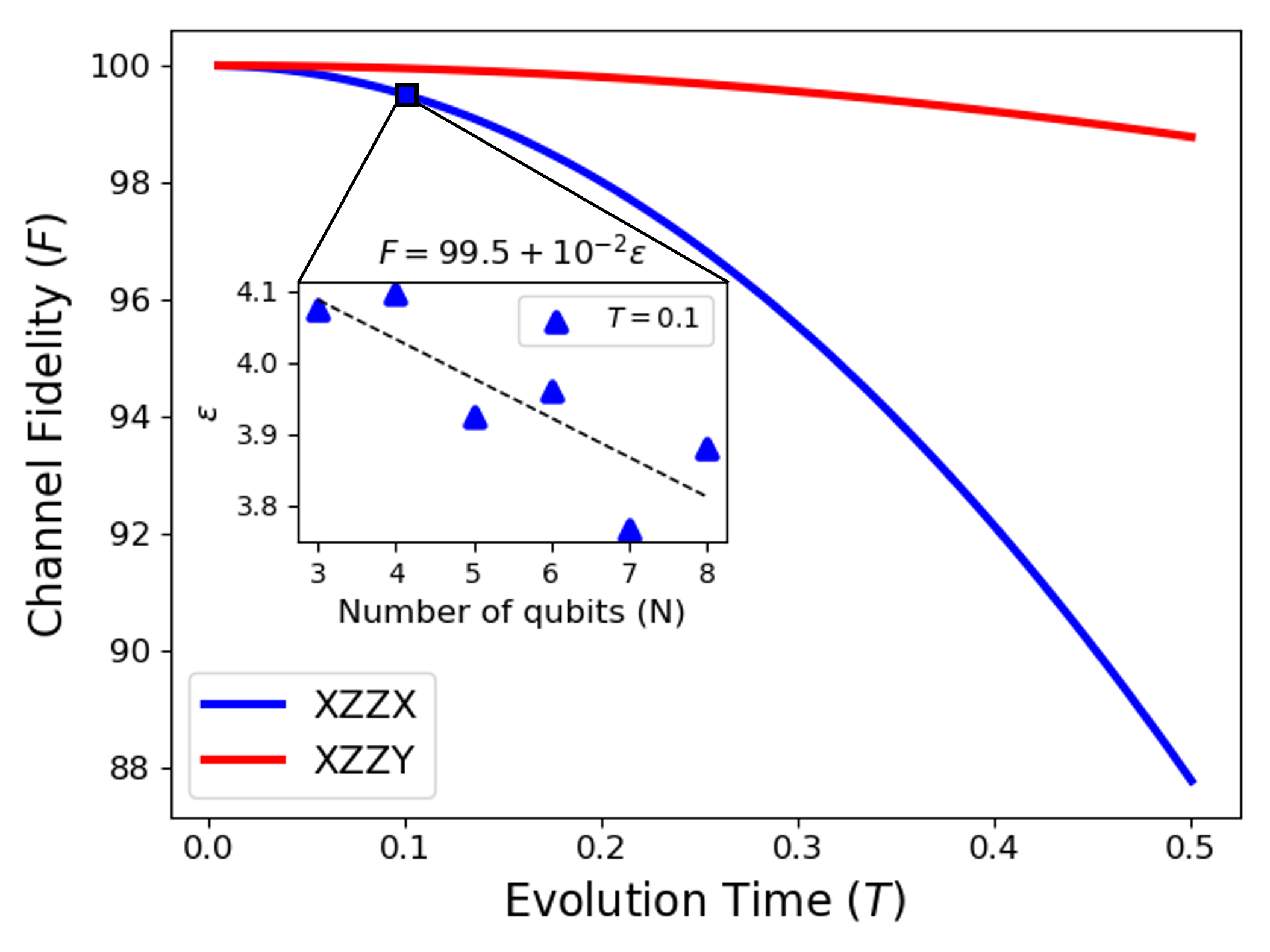}
    \caption{Results of simulating evolution under Jordan-Wigner products on an Ising ansatz. We used $N=4$ qubits to simulate the unitary evolution $e^{-iTP}$ where $P$ is a Jordan-Wigner product as defined in Eq.~\ref{eq:JWProduct} for a total time $T$. In particular, we show the average channel fidelity ($F$) of $XZZX$ and $XZZY$ evolution. We note that both products are faithfully simulated for small $T$. The inset shows the average fidelity obtained $T=0.1$ as a function of the number of qubits ($N$). We observe that the minimal fidelity remains constant around $99.5\%$. For all experiments, $K=10$ control field terms were used. Thus, a high-fidelity simulation is obtained with constant resources in $N$.}
    \label{fig:JWStrings}
\end{figure}
Lastly, this subsection shows how the Magnus expansion can be used to co-design algorithms to simulate the evolution products of Pauli strings, which show up repeatedly in quantum chemistry applications. \par 

Quantum computers promise to help calculations related to chemical reactions and molecular configurations. One approach to achieve this promise relies on mapping $N$ chemically relevant spin orbitals to $N$ qubits with $N-m$ occupied and $m$ virtual orbitals. Then, a reference state a reference state $|\psi_{\text{ref}}\rangle$ is evolved via a unitary $U$ to minimize the energy of $U|\psi_{\text{ref}}\rangle$ as as measured by the molecular Hamiltonian $H_{\text{mol}}$. This approach underpins algorithms like variational quantum eigensolvers (VQEs) \cite{peruzzo2014variational}, where a circuit with variational parameters is tuned for energy minimization regardless of whether the circuit's operation directly affects the occupancy of virtual orbitals. Chemical intuition has resulted in designing $U$ to correlate virtual and occupied orbitals~\cite{greene2022modelling}. This approach underpins unitary couple-clusters (UCC), which aims at evolving under the UCC unitary $U = \exp\{-i\hat{T}\}$ where $\hat{T}$ is the so-called cluster operator defined by 
\begin{equation}
    \hat{T} = \sum_{a,i} c^a_{i} a^\dagger_{a}a_i+\frac{1}{2}\sum_{ab,ij} c^{ab}_{ij} a^\dagger_{a}a^\dagger_{b}a_{i}a_{j}+\hdots. 
\end{equation}
A larger overlap with the ground state is then achieved by tuning the real parameters $\{c^a_{i}, c^{ab}_{ij},...\}$ which mediates coupled occupancy (unoccupacies) of the virtual (occupied) orbitals. \par 
To realize the UCC, it is crucial to implement the UCC unitary on a quantum computer. One popular approach is to break down $U$ using Trotterization \cite{romero2018strategies} and implement each operator coupling operator independently. To do this, we must map operators like $a^\dagger_aa_i$ to the qubit basis using the Jordan-Wigner transformation $a^\dagger_m = \frac{1}{2}\left(\prod_{a<m}Z_a\right)\left(X_m-iY_m\right).$ Notice that, for example, the coupled operator $a^\dagger_aa_i$ necessitates implementing unitary evolutions generated by the following operators, which we call Jordan-Wigner products:
\begin{align}\label{eq:JWProduct}
    &X_a Z_{a+1}\hdots Z_{i-1} X_i, \quad  X_a Z_{a+1}\hdots Z_{i-1}Y_i, \notag \\
    &Y_a Z_{a+1}\hdots Z_{i-1} X_i, \quad Y_a Z_{a+1}\hdots Z_{i-1} Y_i.
\end{align}
Using the Magnus expansion, we see that an AQML algorithm can simulate evolution under Jordan-Wigner products. Take, for example, the Jordan-Wigner product $X_1 Z_{2}\hdots Z_{N-1}X_N$. We can expect that evolution under this operator can be simulated using a rotated Ising ansatz with $H_{\text{nat}} = J \sum_{i}X_i X_{i+1}$ and at least $z$ controls. To see this, consider the Magnus expansion of this model. Through nested commutation, we have that
\begin{align}
    &X_{1}X_{2}\xrightarrow[]{Z_{2}} X_1Y_2 \xrightarrow[]{X_2X_3} X_1Z_2X_3 \xrightarrow[]{Z_{3}} \hdots \notag \\
    &\xrightarrow[]{X_{N-1}X_N} X_1 Z_{2}\hdots Z_{N-1}X_N
\end{align}
The operators above the arrows are the ansatz's Hamiltonian operators with which we commute. It is important to note that the Jordan-Wigner approximation was first proposed as a methodology to solve the time-independent 1D Ising model. Therefore, it is reasonable that the 1D Ising model can also simulate the evolution of Jordan-Wigner products.\par

Fig.~\ref{fig:JWStrings} shows the results of simulating the evolution of different Jordan-Wigner products. We minimized Eq.~(\ref{eq:LossE}) for these numerical experiments, starting from 100 random initialization of control fields and coupling constants $J$ with $K=10$ for $N=4$. Fourier functions were used, but similar results were observed with Legendre and Gaussian functions. In Fig.~\ref{fig:JWStrings}, we show the channel fidelity of the optimized evolution $\U$. We see that both products' evolution was faithfully simulated early, while it is easier to simulate $XZZY$ evolution over a longer period. The inset shows the results of simulating the product $XZ\hdots ZX$ for different qubits for $T=0.1$. We see that the fidelity stays largely constant. Thus, the simulation is faithful with constant resources in $N$. This contrasts with other methods for calculating these products using fermionic swap operators, where faithful simulation requires resources to scale linearly in $N$ \cite{kivlichan2018quantum}. \par 

\section{Conclusion and Outlook}\label{Sec:Conclusion}
In this work, we systematically studied the landscapes of several AQML algorithms. We showed that AQML algorithms violate a key assumption from QOC theory that would suffice for trap-free landscapes. However, we observe that ans\"atze co-designed with specific tasks often showcase trap-free landscapes. Therefore, co-design is paramount to the development of successful AQML algorithms. We show that in the case of time evolution simulation, co-designed can be realized by studying the Magnus expansion of a given ansatz and showcase our approach at tasks relevant to the quantum information community. \par 

This work elucidates various lines of inquiry worth further exploration. Firstly, further work is needed to design an AQML ansatz that satisfies local surjectivity. As pointed out in Ref.~\cite{russell2018reply}, QOC theory's assertion of trap-free landscapes does not apply when the control fields are singular (i.e., they violate local surjectivity). However, recent work shows that even singular ans\"atze depending on powers of control functions can still produce trap-free landscapes \cite{russell2021quantum}. So far, these ansatz's Hamiltonians include higher-order terms resulting from the higher-order effects of shining intense lasers onto quantum systems. Along this direction, further work could determine whether higher-order effects in the envisioned hardware of AQML algorithms can mitigate traps. For example, it is well-known that optical-tweezer arrays of Rydberg atoms exhibit nonlinear dynamics when closely packed due to the Rydberg blockage mechanism. These dynamics could then elucidate trap-free landscapes. Alternatively, classes of basis functions may exist for which local surjectivity is met. \par 
Secondly, our study of the Magnus expansion as a means to co-design AQML algorithms is limited. Based on our work, we hope the community is inspired to study the viability of simulating higher-order terms from the UCC operator.\par 

Lastly, using the operators uncovered by the Magnus expansion to simulate novel quantum dynamics seems particularly enticing. For example, the operator highlighted in Sec.~\ref{Sec:Squeezing} has a precise application to quantum metrology as it creates spin-squeezing. Moreover, such an all-to-all Ising model has recently exhibited a dynamical phase transition \cite{corps2023mechanism}, which can enable new understandings of out-of-equilibrium phenomena \cite{heyl2018dynamical}. Multi-body terms such as those in Sec.~\ref{Sec:JordanWigner} also appear when considering the dynamics of qubits to phononic baths~\cite{gambetta2020long}, and understanding these models may enable long-lived quantum information storage, kinetically constrained dynamics, and correlated quantum states of matter. \par

\begin{acknowledgments}
The authors thank Christian Arentz for their insightful discussion on QOC theory and Nishad Maskara for pointing out applications relevant to quantum chemistry. RAB acknowledges support from the National Science Foundation (NSF) Graduate Research Fellowship under Grant No. DGE1745303. JGP acknowledges support from the Harvard College Research Program and the Harvard Quantum Initiative. SFY acknowledges funding from NSF through the Q-IDEAS HDR Institute (OAC-2118310), the Q-SeNSE QCLI (OMA-2016244), and the CUA PFC (PHY-2317134). The authors also acknowledge the funding through the DARPA IMPAQT Program (HR0011-23-3-0023).

\end{acknowledgments}
\bibliographystyle{unsrtnat}
\bibliography{bibliography}

\onecolumngrid
\appendix

\section{QP's Expressivity and the Ising Evolution}\label{A:QPs_IsingEvolution}

The control fields $f_i^\alpha$ are sums of simple functions $f_i^\alpha(t;\bs{a}) = \sum_{k=1}^Ka_{ik}^\alpha g_k(t)$. We use the following different functions,
\begin{align}
    g_k(t) &= \cos\left(\frac{2\pi k t}{T}\right) \text{ (Fourier) } \label{eq:Fourier}\\
    g_k(t) &= \exp\left(-\frac{(t-Tk/K)^2}{2\sigma^2}\right) \text{ (Gaussian) } \label{eq:Gaussian} \\
    g_k(t) &= \text{LegendrePolynomial}(t/T, k) \text{ (Legendre) } \\  \label{eq:Legendre}
    g_k(t) &= \Theta(kT/K\leq t\leq (k+1)T/K)\text{ (Piece-wise constant) } \\  \label{eq:PWC}
\end{align}
where $T$ is the total evolution time and $\sigma$ a pulse-width. The Fourier expansion is continuous, while the Gaussian expansion more closely represents a circuit-based (or pulse-based) system. \par 

We now show that a QP with Gaussian controls and piece-wise constant controls on the inputs can simulate the target unitary in Eq.~\ref{eq:IsingEvolution}. To do so, let us first observe that we can Trotterize the Ising evolution in Eq.~\ref{eq:IsingEvolution} such that 
\begin{equation}
    U^{\text{targ}} \approx \left(e^{-i\frac{1}{n}\sum_i Z_iZ_{i+1}}e^{i\frac{0.1}{n}\sum_i X_i}\right)^n+\mathcal{O}(1/n).
\end{equation}
Thus, we only have to ensure that a QP ansatz can produce the unitaries $U^i_R = e^{-i\alpha X_i}$ and $U^i_I = e^{-i\beta Z_iZ_{i+1}}$ for every $i$ and for, at least, small $\alpha, \beta.$ \par 
Consider a QP with Gaussian pulses of width $\sigma\ll 1.$ Suppose the Gaussian functions $g_k(t)$ have a unit height and are centered at $t=0, T/K, 2T/K,..., (K-1)T/K$. The parameters $a_{ik}^{x,y}$ paired with these Gaussian functions dictate the magnitude of the control field. Additionally, suppose the output qubit contains piecewise constant functions that change at the same intervals. The magnitude of the PWC functions is mediated via parameters $b_{Nk}^{x,y}$. In that case, one can faithfully break down the unitary evolution of a QP as follows: 
\begin{equation}
    U(\bs{\theta})\approx \prod_{k=0}^{K-1} \left(e^{-i\frac{T}{K}\left(\sum_{i=1}^{N-1} J_i Z_iZ_{N}+b_{Nk}^xX_N+b_{Nk}^yY_N\right)}\prod_{i=1}^{N}e^{-i\sigma \left(a_{ik}^x X_i + a_{ik}^y Y_i\right)}\right).
\end{equation}
Immediately, we notice that the right-most term of the equation above contains the rotations $U^i_R$. Moreover, the left-most term contains the unitary $U^{N-1}_I$. However, we are yet missing Ising-type interactions between different input qubits (i.e., $Z_{i}Z_{i+1}$ for $i<N-2$). Suppose, however, that the PWC parameters are such that $J/b\ll 1$. For the sake of argument, suppose that only the $x$-field is present. In that case, we have already shown in Ref.~\cite{bravo2022quantum} that this effectuates the following evolution on the input qubits: 
\begin{equation}\label{eq:Effective}
    U_{I,\text{eff}}^{k}= e^{-i\frac{T}{K}\sum_{ij}\frac{J_iJ_j}{b_{Nk}^x}Z_iZ_j}
\end{equation}
This contains the interaction terms $U_{I}^{i}$ for which $j=i+1$. However, it also contains unwanted interactions, which we want to eliminate. This can be done by choosing the two subsequent rotations generated by the Gaussians be of opposite sign and large magnitude such that $\sigma a_{jk}^{x} = \pi = -\sigma a_{j(k+1)}^{x}$ for a $j$ which we do not want. This means that per interaction term $U_I^i$ that we want to produce, we must ensure we use at least three values of $k$: one used to rotate qubits according to $U_{R}^i$, and two with $b_{Nk}^x$ large and equal to make sure that for all $j\neq i+1$, we have a vanishing interaction. Notice that this works because 
\begin{equation}
    e^{-i\frac{J_iJ_j T}{Kb_{Nk}^x}Z_iZ_j}e^{i\pi X_j}e^{-i\frac{J_iJ_j T}{Kb_{Nk}^x}Z_iZ_j}e^{-i\pi X_j} = e^{-i\frac{J_iJ_j T}{Kb_{Nk}^x}Z_iZ_j}e^{i\frac{J_iJ_j T}{Kb_{Nk}^x}Z_iZ_j}=\mathds{1}. 
\end{equation}
For the case of $j=i+1$, the total evolution time of the interaction is $\beta = 3J_{i}J_{i+1}T/Kb_{Nk}^x$, which is necessarily small. \par 
This procedure necessitates that $K$ be large enough for each interaction to be tuned independently. This means that $K \geq 3N n$ where $n$ is the number of Trotter steps. \par 
Moreover, our proof heavily relies on producing the effective Hamiltonian in Eq.~\ref{eq:Effective}. Therefore, it is paramount that the QP ansatz includes PWC functions on the output qubit. In numerical experiments, we have seen that while Gaussian controls work best, Fourier controls also allow the QP to simulate the Ising evolution as long as the output contains PWC controls. We say that a QP with PWC controls is an A2 ansatz as specified in Sec.~\ref{Sec:Groups}.\par 

\section{Landscapes of unitary time evolution}\label{A:LE_Landscapes}

\subsection{Derivation of the Dynamical Derivative}\label{A:Dynamical}
Assume that we can break the evolution of our system into $M$ constant portions: 
\begin{equation}
    U = \prod_{i=1}^M\exp\left(-i\delta t_i\left(H_{\n}+H_{\ctr}(\delta t_i/2)\right)\right)
\end{equation}
where $\delta t_i = t_i-t_{i-1}$. Then, 
\begin{align}
    \frac{\partial U}{\partial \theta^\alpha_{nk}} &= \sum_{i=1}^M\prod_{j=i+1}^M e^{-iH(\delta t_j/2)\delta t_j}\frac{\partial}{\partial\theta_{nk}^\alpha} e^{-i\left(H_\n+\sum_{nk\alpha}\theta^\alpha_{nk}g_n(\delta t_i/2)S_n^{\alpha}\delta t_i\right)}\prod_{j=1}^{i-1} e^{-iH(\delta t_j/2)\delta t_j}\\
    &=-i\sum_{i=1}^M\prod_{j=i}^M e^{-iH(\delta t_j/2)\delta t_j} g_k(\delta t_i/2)S_{n}^\alpha\delta t_i \prod_{j=1}^{i-1} e^{-iH(\delta t_j/2)\delta t_j}.
\end{align}
Taking the continuous limit and using the midpoint definition of integration, we obtain 
\begin{equation}
    \frac{\partial U}{\partial \theta^\alpha_{nk}} = -i\int_{0}^{T} g_k(t) U(T,t)S^{\alpha}_n U(t,0)dt.
\end{equation}
By introducing the unity $1=U(\boldsymbol{\theta})U^\dagger U(\boldsymbol{\theta})$ on the left side of the integral one obtains Eq.~\ref{eq:dynamical_derivative}. 

\subsection{Trap-free Landscapes}
For this section, we use the loss $\mathcal{L}_E = ||(U-W)||^2_F$ where $||A||_F^2=\mathrm{Tr}(A^\dagger A)$ is the Fr\"obenius norm. Note that this loss differs from that in Eq.~\ref{eq:LossE} by a factor of $2^{N+1}$, but this discrepancy only changes its values, not its curvature. Evidently, 
\begin{equation}
    \mathcal{L}_E = 2(2^N) - \tr\left(W^\dagger U + U^\dagger W\right)
\end{equation}
where we use the notation $W=U^{\text{targ}}$. Let us express the dynamical derivative in Eq.~\ref{eq:dynamical_derivative} as $\partial_{\theta_{nk}^\alpha}U = -i U \mu_{nk}^\alpha$. \par
The derivative of the loss function is 
\begin{equation}
    \frac{\partial \mathcal{L}}{\partial \theta_{nk}^\alpha} = -i\tr\left((U^\dagger W-W^\dagger U)]\mu_{nk}^\alpha \right).
\end{equation}
If we assume local surjectivity, then a critical point where all derivatives vanish can only be accomplished if $U^\dagger W = W^\dagger U$ \cite{ho2009landscape}. Obviously one answer is that $U = e^{i\phi} W$ is a critical point for any phase $\phi$. We denote $U^\dagger W$ as $\chi$. Note that $\chi^2 = I$. At the critical points, $\chi=\chi^\dagger$ and so it must have eigenvalues $\pm 1$ ($\phi=0,2\pi$). For the critical points, the loss function takes on the following values: 
\begin{equation*}
    \mathcal{L}_E(\text{ c.p. }) = 2\left(2^N-\sum_{i=1}^{2^N}(-1)^{n_i}\right)
\end{equation*} 
where $n_i=0$ if the $i^{th}$ eigenvalue of $\chi$ is positive and $n_i=1$ otherwise. Clearly, the case where $W=U$ yields $\mathcal{L}_E=0$ and $W=-U$ yields $L_{E}=4(2^N)$. Other critical points can be found by taking $W$ and permuting $n$ columns using a permutation matrix $\Pi$. In such case, $\chi = \Pi^\dagger$. Evidently, the loss becomes $\mathcal{L}_E=2(2^N-n)$ which has a degeneracy of $\binom{2^N}{n}$. \par 
Let's now calculate the Hessian of the loss function $\LE$. We obtain
\begin{align}
    \frac{\partial^2\LE}{\partial \theta^\alpha_{nk}\theta^\beta_{ml}} &= \mathrm{Tr}\left(\chi \mu_{nk}^\alpha \mu_{ml}^\beta + \chi^\dagger \mu_{ml}^\beta\mu_{nk}^\alpha \right) -i\mathrm{Tr}\left((\chi-\chi^\dagger)\frac{\partial \mu_n^k}{\partial\theta_m^l}\right)
\end{align}
since $\mu_{nk}^\alpha$ is hermitian.\par 
Note that 
\begin{align}
    \frac{\partial S_{n}^\alpha}{\partial \theta_{ml}^\alpha} &= \frac{\partial U^\dagger}{\partial \theta_{ml}^\beta}S_n^\alpha U(t) + U^\dagger S_n^\alpha \frac{\partial U}{\partial \theta_{ml}^
    \beta} \notag \\
    &= i[\mu_{ml}^\beta,S_n^\alpha(t)].\notag
\end{align}
Therefore, 
\begin{equation}
    \frac{\partial \mu_{nk}^\alpha}{\partial \theta_{ml}^\beta} = i\int_0^Tdt_1\int_{0}^{t_1}dt_2 g_k(t_1)g_l(t_2) [S_m^\beta(t_2), S_n^\alpha (t_1)] = i\int_0^Td\tau g_k(\tau)[\mu_{ml}^\beta(\tau),S_n^\alpha(\tau)].
\end{equation}
At a critical point, $\chi^\dagger=\chi$, and therefore 
\begin{align}\label{eq:HessianCP}
    \frac{\partial^2\LE}{\partial \theta^\alpha_{nk}\theta^\beta_{ml}}(\text{c.p.}) &= \mathrm{Tr}\left(\chi \{\mu_{nk}^\alpha, \mu_{ml}^\beta\} \right).
\end{align}
Since $\chi$ is diagonalizable, there exist a change of basis $D$ such that $\chi = D A D^\dagger$ where $A=\text{diag}((-1)^{n_1},(-1)^{n_2},...,(-1)^{n_d})$ with $d=2^N.$ Let $a,b$ stand for the indexes of $\theta_{nk}^\alpha$ and $\theta_{ml}^\beta$. Eq.~(\ref{eq:HessianCP}) is the entry $\mathcal{H}_{ab}$ of the Hessian $\mathcal{H}$. Inserting the identity $DD^\dagger$ into Eq.~(\ref{eq:HessianCP}) we obtain 
\begin{equation}
    \mathcal{H}_{ab} =\tr\left(A(D^\dagger\mu_a DD^\dagger \mu_b + D^\dagger\mu_n DD^\dagger \mu_a D)\right).
\end{equation}
Let's call $\bar{\mu}_{a,b} = D^\dagger \mu_{a,b}D$. Then, 
\begin{align*}
    \mathcal{H}_{ab} &= \sum_{i=}^d(-1)^{{n_i}}\sum_{j=1}^{d}\left( \langle i|\bar{\mu}_a|j\rangle\langle j|\bar{\mu}_a|i\rangle + \langle i|\bar{\mu}_b|j\rangle\langle j|\bar{\mu}_a|i\rangle\right) \\
    &=2\sum_{i=1}^d (-1)^{n_i}\langle i|\bar{\mu}_a|i\rangle\langle i |\bar{\mu}_b|i\rangle \\
    &+2\sum_{j<i}\left((-1)^{n_i}+(-1)^{n_j}\right)\left(\text{Re}[\langle i|\bar{\mu}_a|j\rangle]\text{Re}[\langle i|\bar{\mu}_b|j\rangle]+\text{Im}[\langle i|\bar{\mu}_a|j\rangle]\text{Im}[\langle i|\bar{\mu}_b|j\rangle] \right).
\end{align*}
Therefore, we have that 
\begin{align}
    \mathcal{H}_{ab} &= \bs{\mu}_a^T \bs{\Gamma} \bs{\mu}_b\\
    \bs{\mu}_a^T &= (\text{Re}[\langle i|\bar{\mu}_a|j\rangle], \text{ } 1\leq i,j\leq d, \text{Im}[\langle i|\bar{\mu}_a|j\rangle] \text{ } 1\leq i < j\leq d)
\end{align}
with the $d^2\times d^2$ matrix $\bs{\Gamma}$ defined as 
\begin{equation}\label{Eq:Gamma}
    \bs{\Gamma} = 2\text{diag}\left\{(-1)^{n_1},\hdots, (-1)^{n_d}, (-1)^{n_i}+(-1)^{n_j} \text{ } 1\leq i <j \leq d, (-1)^{n_d}, (-1)^{n_i}+(-1)^{n_j} \text{ } 1\leq i <j \leq d \right\}.
\end{equation}
So far, this theory has yielded the same results as in Ref.~\cite{ho2009landscape}, however, while in that case the Hessian is defined as on different times $\mathcal{H}(t,t')$ since the functions the controls are to be tunned at every point in time, in an AQML model, the functions are only tunable through a series of finite parameters $N_{p}$. The entire Hessian is a finite matrix defined by 
\begin{equation}
    \mathcal{H} = M^T \bs{\Gamma} M = \begin{pmatrix}
    \bs{\mu}^T_1 \\
    \bs{\mu}^T_2 \\
    \vdots \\
    \bs{\mu}^T_{N_p}
    \end{pmatrix}\bs{\Gamma}\begin{pmatrix}
    \bs{\mu}_1 & \bs{\mu}_2 \hdots & \bs{\mu}_{N_p}
    \end{pmatrix}.
\end{equation}
The matrix $M$ is a ($d^2\times N_p$) matrix. Since we have assumed local surjectivity, the entries of $\bs{\mu}_a$ are all linearly independent, but not necessarily orthogonal. \par 
When $N_p\geq d^2$, we can apply some of the results in Sec. VI of Ref.~\cite{ho2009landscape}. These results show that there exists an orthonormal basis $\bs{O}$ that diagonalizes $\mathcal{H}$ ($C = O^T\mathcal{H}O$) where $C$ is a diagonal matrix whose entries are the curvature of the landscape and the critical points. Thus, the matrix $\bs{\Gamma}$ is congruent to $C$. Using Sylvester's law of inertia, congruent matrices share rank (i.e., the number of nonzero eigenvalues) and signature (i.e., the number of positive minus negative eigenvalues). A signature of magnitude equal to the rank means the critical point is a maximum or a minimum if positive or negative, respectively. A signature smaller than the rank in magnitude means the critical point is a saddle point. This is because a critical point with such a signature would have at least one direction in which the curvature is of the opposite sign as other directions. Thus, Eq.~(\ref{Eq:Gamma}) tells us the most relevant information about the curvature by looking at the sign of its entries and the number of zero entries. \par 
For the case that $U=W$ $n_i = 0$ and thus, the fist $d$ entries of $\bs{\Gamma}$ are positive. The rest of the entries are zero. Thus, the rank of the Hessian is $d$ and its signature $d$. Thus, this critical point is a global minimum. For the case that $U=-W$, $\Gamma$ has a rank of $d$ and a signature of $-d$, and it is thus a global maximum. For the cases that $U=\Pi W$ where $\Pi$ is a permutation of $n$ columns, the rank is $R_n= d(d-2n)+2dn^2$ greater than $d$ for $n>0$. The signature is $S_n= d(d-2n)$. Clearly, for $n>0$ $S_n<R_n$.\par 
Let's summarize our findings. We made two significant assumptions: 
\begin{enumerate}
    \item We assumed that our AQML algorithm is locally surjective. In practice, this needs to be checked at least numerically. 
    \item We assumed that we have at least $2^{2N}$ variational parameters. In practice, this is a choice experiments must decide on. Making this choice can be given the meaning of choosing to overparametrize the AQML algorithm.
\end{enumerate}
Given those assumptions, we found that 
\begin{enumerate}
    \item The critical points of $\LE$ correspond to the cases in which $U^\dagger W = W^\dagger U$.
    \item The values of the loss function at these critical points is an integer given by $2(2^N-n)$ where $0\leq n \leq 2^n$ is an integer. Each critical point has a degeneracy of $\binom{2^N}{n}$ solutions. \par 
    \item The rank of the Hessian at the critical points is $R_n= d(d-2n)+2dn^2$ which achieves a maximum value of $d^2$ for $n=2^N$ ($U=W$) or $n=0$ ($U=-W$). 
    \item The signature of the Hessian at the critical points is $S_n= d(d-2n)$. Thus, the critical point of the case of $W=U$ ($n=0$) is a global minimum. That is, the Hessian contains only non-negative eigenvalues. Similarly, $W=-U$ ($n=d$) is a global maximum. That is, the Hessian contains only non-positive eigenvalues. For all other critical points, there is at least one positive and one negative eigenvalue and thus they are saddle points.
\end{enumerate}
In conclusion, the assumptions of local surjectivity and overparametrization should lead to trap-free landscapes of $\LE$. 

\section{AQML Models and Local Surjectivity}\label{A:LSurj}
This appendix shows that AQML ansatz fails to satisfy local surjectivity in certain cases. To prove this statement, it is helpful to establish the following lemmas. 

\begin{lemma}\label{lemma:lemma1}
Suppose $h$ and $P$ are two Pauli strings and that $U_1$ and $U_2$ are independent unitary 1-designs. Then,
\begin{equation}\label{eq:Lemma1}
    \mathbb{E}_{U_1, U_2}\left(\mathrm{Tr}(P U_1 U_2^\dagger h U_2)\right)=0.
\end{equation}
\end{lemma}
\begin{proof}
    First, recall that for a 1-design $\mathbb{E}_U(U^\dagger A U) = \frac{\mathrm{Tr}(A)}{2^{N}}\mathds{1}.$ Then, 
    
    \begin{align}
        \mathbb{E}_{U_1, U_2}\left(\mathrm{Tr}(P U_1 U_2^\dagger h U_2)\right) &= \mathbb{E}_{U_1}\left(PU_1 \mathbb{E}_{U_2}\left(U_2^\dagger h U_2 \right) \right) \notag\\
        &=\frac{\mathrm{Tr}(h)}{2^N}\mathbb{E}_{U_1}\left(PU_1 \right) = 0.
    \end{align}
    This follows from $\mathrm{Tr}(h)=0$.
\end{proof}

\begin{lemma}\label{lemma:lemma2}
    Suppose $M$ is a Hermitian operator and $U$ is a Haar random matrix. Then, 
    \begin{equation}
        \mathbb{E}_U(\mathrm{Tr}\left(UM\right))= 0.
    \end{equation}
\end{lemma}

\begin{proof}
    Since $M$ is Hermitian, it can be diagonalized. That is, a unitary $D$ exists for which $M=D\Lambda D^\dagger$ where $\Lambda$ is diagonal. Then, using the right and left invariances of the Haar measure 
    \begin{align*}
        \mathbb{E}_U(\mathrm{Tr}\left(UM\right)) &= \mathrm{Tr}\left( \int_U dU U D \Lambda D^\dagger\right)\\
        &=\mathrm{Tr}\left( \int_U dU U \Lambda D^\dagger\right)\\
        &= \mathrm{Tr}\left( \int_U dU D^\dagger U \Lambda D^\dagger\right) \\
        &= \mathrm{Tr}\left( \int_U dU  U \Lambda D^\dagger\right) \\
        &= \sum_{i=1}^d\int_{U}dU U_{ii}\Lambda_i.
    \end{align*}
    By symmetry, each of the terms in this last equation should vanish since if they didn't, then the eigenvector associated with the eigenvalue $\Lambda_i$ would be a special direction for the Haar random unitary. 
\end{proof}

Using these lemmas, we can now make a substantive statement about local surjectivity for AQML ans\"atze. Notice that testing local surjectivity amounts to calculating the following quantity

\begin{equation}
    \mathbb{E}_{\bs{\theta}}\left(\mathrm{Tr}\left(P\partial_{\bs{\theta} U}\right)\right) = -i\int_{0}^Td\tau g_k(t) \mathbb{E}_{\bs{\theta}}\left(PU(T;\bs{\theta})U^{\dagger}(\tau; \bs{\theta})hU(\tau; \bs{\theta})\right).
\end{equation}
Here, we make a few extra assumptions. We assume that $U(\tau)$ and $U(T)$ are independent for $\tau<T$, that $U(\tau)$ is a 1-design for $\tau>0$, and that $U(T)$ is a Haar random unitary. These assumptions imply that $U(T)$ is also a 1-design. For $0<\tau<T$, the integrand vanishes due to Lemma~\ref{lemma:lemma1}. The edge case $\tau=0$ give us an integrand containing $\mathbb{E}_{\bs{\theta}}(\mathrm{Tr}(U(\bs{\theta}) h P) )$. The edge case $\tau=T$ gives an integrand containing $\mathbb{E}_{\bs{\theta}}(\mathrm{Tr}(U(\bs{\theta}) P h) )$. Since Paulis are closed under multiplications, both cases have integrands of the form of Lemma~\ref{lemma:lemma2} and, therefore, vanish. \par 

We can, therefore, conclude that local surjectivity will be violated for every Pauli string if an AQML algorithm such that $U(\bs{\theta})$ is a Haar random unitary. Notice that this statement also assumes that $U(\tau)$ and $U(T)$ are independent for $\tau<T$, that $U(\tau)$ is a 1-design for $\tau>0$. \par 

Here, we must clarify that one could also relax the Haar random requirement if $g_k(t)$ vanishes for $t\approx T$ and $t\approx 0$. In that case, an AQML algorithm producing unitaries that only resemble 1-designs suffices for the violation of local surjectivity. \par 

Another caveat to our proof is that the Haar random requirement of Lemma~\ref{lemma:lemma2} is stringent. It remains to be shown whether the average $\mathrm{Tr}(UM)$ can vanish for other kinds of unitary ensembles. \par 

\section{Analog Landscapes}\label{A:Traps}

We numerically simulated the time evolution of a 1D Ising with a transverse field defined by  
\begin{equation}
    U^{\text{targ}} = \exp\left\{-i\left(\sum_{i=1}^{N-1}Z_iZ_{i+1}+h\sum_{i=1}^N X_i\right)\right\}.
\end{equation}
with $h=0.1$. We approximated this evolution with $N=2,3,4$, and $5$ qubits using the loss in Eq.~(\ref{eq:LossE}). 

The simulation code was implemented with PennyLane's differentiable pulse programming capabilities~\cite{bergholm2018pennylane} and the JAX framework. The optimizers were similarly implemented using the Optax library. We note that these libraries/frameworks support automatic differentiation; thus, the results are exact up to machine precision.

The first phase of the numerical experiments consists of training $100$ randomly initialized A1 QP ansatz for each qubit number. These are sampled from a uniform distribution on the interval $[0,1]$. We parameterize the pulses as a sum of the first $K$ terms of a Fourier series, where $K = 5N$ for $N$ qubits. Thus, concretely, we used $10,15,20$, and $25$ terms for each control pulse for the experiments with $2,3,4$ and $5$ pulses, respectively. We note that the $K=5N$ cutoff for the series is arbitrary, but the system's behavior remained unchanged for several different $K$. As for the optimizer, we used Optax's Adam implementation.

\begin{figure}[h]
    \centering
    \includegraphics[width=0.8\linewidth]{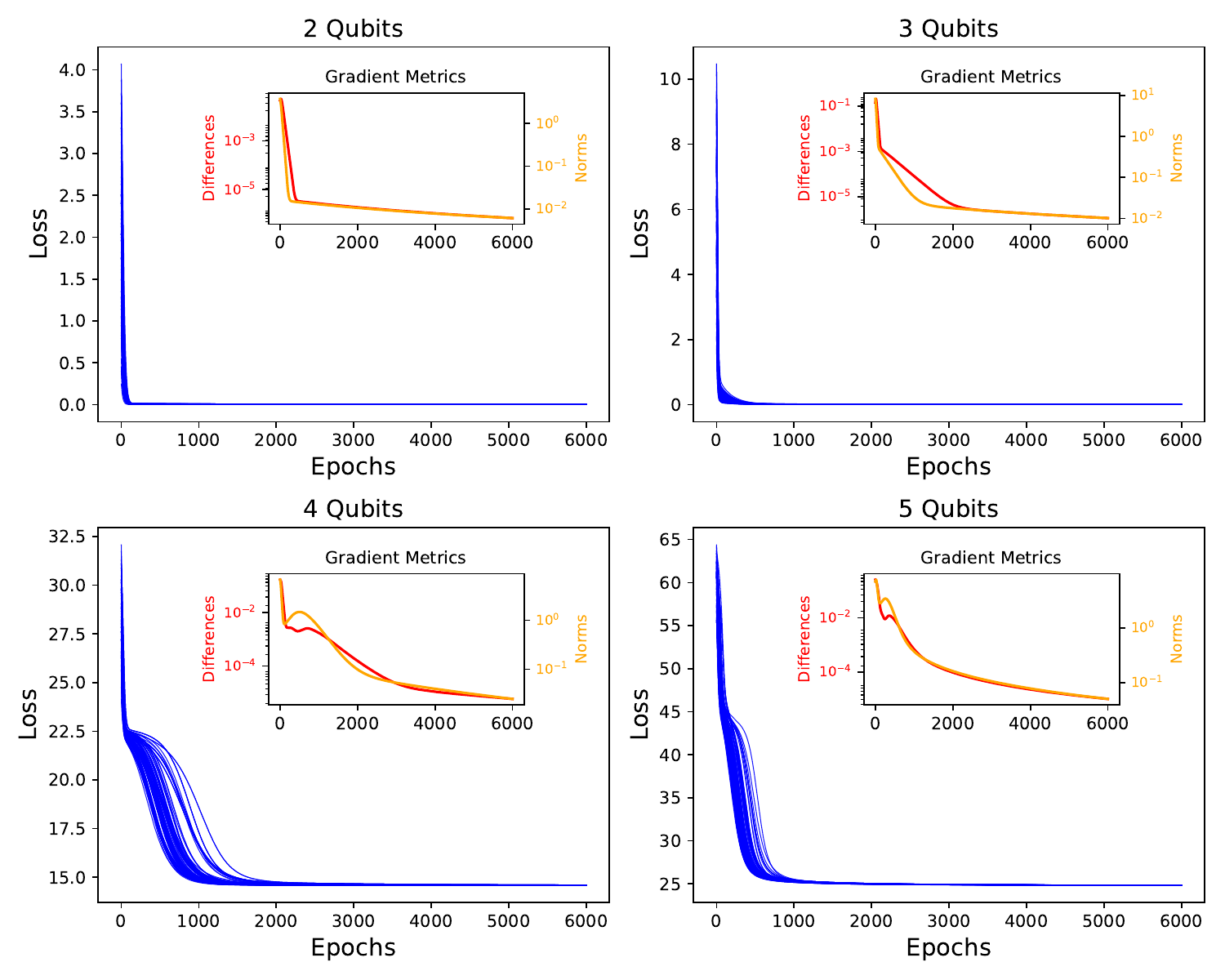}
    \caption{Training history of $100$ random A1 QP ansatz on $2,3,4$ and $5$ qubits. The QPs were trained for $6000$ epochs and we observe that for each trial $\mathcal{L}_E$ plateaus, suggesting that the training converged to a critical point. For $2$ and $3$ qubits, all trials converge to the optimal value of $\mathcal{L}_E = 0$, yet for $4$ and $5$ this is no longer the case and they all converge to sub-optimal values. The inset shows logarithmic plots of two different gradient metrics used to confirm convergence: the $L^2$ norm of the gradient and the $L^2$ norm of the differences between gradients at consecutive epochs. These metrics further confirm convergence to a critical point.}
    \label{fig:QP_Training_History_all}
\end{figure}

Fig.~\ref{fig:QP_Training_History_all} shows the results of the numerical experiments. The training process is successful for all cases and converges to a critical point. Yet the optimal approximated unitary is only achieved at a value of $\mathcal{L}_E = 0$. Thus we see that for $4$ and $5$ qubits, all sampled QP ansatz converge to a sub-optimal critical point (i.e., a local minimum or maximum instead of the global optimum). We see similar results with the A2 QP ansatz with a lower loss for $N\geq 4$.
 
To confirm that the training had indeed converged, we focus on two additional metrics: the mean $L^{2}$ norm of the gradient at each point Eq.~(\ref{eq:average_norm_gradients_app}) and the mean $L^{2}$ norm of the differences between the gradients at epoch $i$ and epoch $i-1$ Eq.~(\ref{eq:average_norm_gradients_differences_app}):

\begin{align}
&\frac{1}{N_{\text{trials}}}\sum_{j=1}^{N_{\text{trials}}} \left\Vert \nabla \mathcal{L}_E^{(i)} \right\Vert^2 \label{eq:average_norm_gradients_app} \\
&\frac{1}{N_{\text{trials}}}\sum_{j=1}^{N_{\text{trials}}} \left\Vert \nabla \mathcal{L}_E^{
(i)} - \nabla \mathcal{L}_E^{(i-1)} \right\Vert^2
\label{eq:average_norm_gradients_differences_app}.
\end{align}

To make sense of these metrics, recall that at a critical point of $\mathcal{L}_E$, its gradient $\nabla \mathcal{L}_E = 0$ by definition. Thus, the mean $L^{2}$ norm of the gradient will only be zero at a critical point, and we would expect the gradients at that epoch and onwards to remain very small and close to $0$. Similarly, because the gradients are zero at a critical point, the update rule of the optimizer at the critical point wouldn't change the current parameters. Thus, the difference between the gradients at epoch $i$ and $i-1$ at a critical point should remain small and close to $0$.

We see that for all experiments, both metrics keep decreasing and plateauing by the end of the training. For $2$ and $3$ qubits, these plateau occurs at approximately $10^{-5}$, whereas for $4$ and $5$ qubits it is located at around $10^{-4}$. 

The second phase of the numerical experiments consists of calculating the Hessian matrix at the critical points found during training. Recall that the Hessian matrix encodes information about the second derivatives of the function at a given point and thus encodes information about the local curvature of the landscape around that point. In particular, a positive eigenvalue corresponds to a direction (the direction of its corresponding eigenvector) with positive curvature, and conversely, a negative eigenvalue corresponds to a direction with negative curvature. A zero eigenvalue represents a direction with zero curvature (i.e., flat).

\begin{figure}[h]
    \centering
    \includegraphics[width=0.8\linewidth]{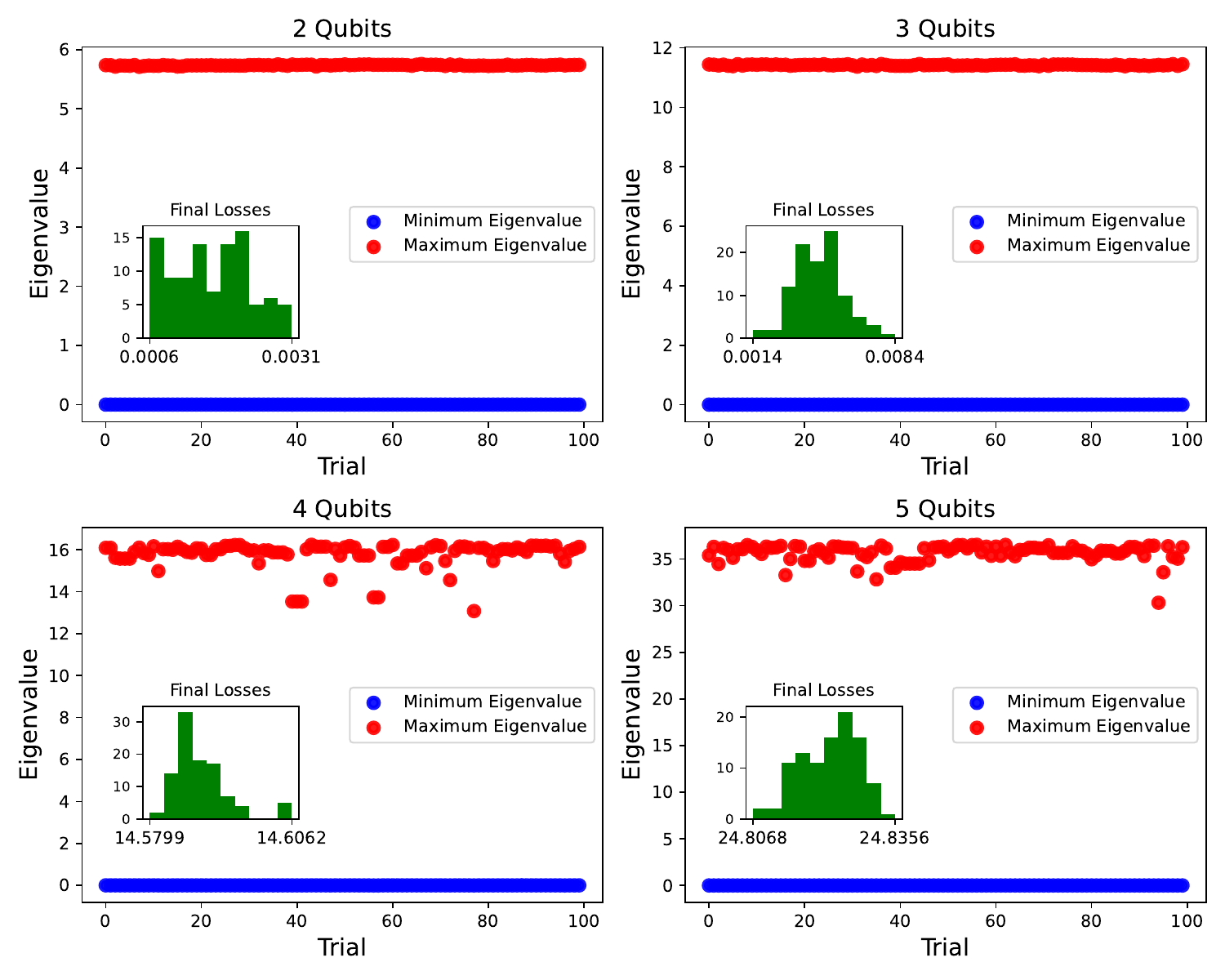}
    \caption{Minimum and Maximum eigenvalues of the Hessian matrix at the converged critical points for each trial of an A1 QP. Since the minimum eigenvalue is always $0$, every trial converged to a critical point with a positive semidefinite Hessian. The presence of a $0$ eigenvalue indicates that in some direction around the critical point, the landscape is locally "flat." The insets show a histogram of the final losses at each critical point.}
    \label{fig:QP_Hessian_Eigenvalues_all}
\end{figure}

Fig.~\ref{fig:QP_Hessian_Eigenvalues_all} shows the minimum and maximum eigenvalues of the Hessian matrices after the training process converges to a critical point for each trial of A1 QPs. The insets show that for $2$ and $3$ qubits, the final loss is indeed optimal. Yet even for this optimal case. We observe that there is always a $0$ eigenvalue corresponding to some locally "flat" direction around the critical points. That is, even for qubit sizes that can converge to an optimal value, the landscape is not as well-behaved as the theory predicts. This suggests that the ability to converge to the global minimum for the A1 QP and similar AQML ans\"atze is not a consequence of a trap-free landscape. We note that for A2 QPs, the eigenvalues are both positive and negative (see Fig.~\ref{fig:a1_a2_hessians}), and so those landscapes are trap-free. 

\section{Algorithm-Task Co-Design}\label{A:Codesign}
\subsection{Magnus Expansion Analysis}
Every unitary $U$ could be written as the evolution under a Hamiltonian $H_U$ via $U = \exp\left(-i T H_U\right)$. Formally, $H_U = \frac{i}{T}\log(U)$. On the other hand, a time-dependent Hamiltonian generates evolution as in Eq.~(\ref{eq:FormalUnitary}). This evolution could be written as $\exp\left(-iT H_{\text{eff}}\right)$ where $H_{\text{eff}}$ is an effective Hamiltonian. \par 
According to the average Hamiltonian theory \cite{brinkmann2016introduction, choi2020robust}, the effective Hamiltonian is given by the Magnus-expansion
\begin{align}
H_{\text{eff}} &= \sum_{l=0} H^{(l)} \\
H^{(0)} &= \frac{1}{T}\int_0^T dt_0 H(t_0)\\ 
H^{(1)} &= \frac{-i}{T}\int_0^T \int_0^{t_1} dt_1dt_0 [H(t_1),H(t_0)]\\ 
H^{(l)} &= \frac{(-i)^l}{l!T}\int_{0}^Tdt_{l}\hdots\int_{0}^{t_1}dt_0 [H(t_l),\hdots [H(t_1), H(t_0))]\hdots].
\end{align}
In the case that $T||H(t)||_S\ll 1$ (where $||A||_S$ is the spectral norm of an operator), it is possible to truncate the expansion since the $l^{th}$ term is of order $\mathcal{O}(T^l||H(t)||_S^l).$ However, all orders are relevant to the cases presented in the main text. \par 
Note that for a constant Hamiltonian, all terms with $l>0$ cancel since $[H(t_1), H(t_2)] = 0.$\par 
The Hamiltonians generated by the Magnus expansion depend on the native and control Hamiltonians. We will use $\mathcal{S}_{H}^l$ to denote the operators generated by the nested commutators in $H^{(l)}$. This means that $H^{(l)}$ is a linear combination of the operators in $\mathcal{S}_{H}^l$ (i.e.$H^{(l)} = \sum_{\mathcal{O}\in \mathcal{S}_{H}^l}\alpha_{\mathcal{O}}\mathcal{O}$). It is worth noting that different orders may produce the same operators. That is, in general, $\mathcal{S}_{H}^l \cap \mathcal{S}_{H}^{m}\neq \emptyset$. An example here is instructive. Consider $H_{\text{nat}} = J Z_1Z_2$ with $x$-controls $H_c(t) = \sum_i f_i^x(t)X_i$. Clearly, $S_H^0$ contains $Z_1Z_2$. We can obtain this operator by commuting $Z_1Z_2$ with the operator $X_i$ twice or with $X_1$ once and then with $X_2$ (or vice-versa). Thus, $Z_1Z_2\in \mathcal{S}_H^{2}$ as well. \par 
Theoretically, one can approximate every unitary evolution $U$ for which $H_U$ is in the span of $\mathcal{S}_H = \cap_l \mathcal{S}_H^l$. In practice, however, the Hamiltonians generated by the first few terms of the Magnus expansion may be easier to optimize and find. \par 
The coefficients in front of each generator are a parametric function. For example, going back to our example, up to second order, 
\begin{align}
    \alpha_{zz} &= J  -\frac{J}{2T}\int_0^T dt_2 \int_0^{t_2}dt_1 \left(f^x_1(t_2)f_1^x(t_1)+ f^x_2(t_2)f_2^x(t_1)+f^x_2(t_2)f_1^x(t_1)+f^x_1(t_2)f_2^x(t_1)\right)t_1.
\end{align}

Notice that we have not simplified the last two terms since the order of integration matters. To see this, suppose the controls are polynomials of orders two and three. Then, note that 
\begin{align}
    \int_0^Tdt_2\int_{0}^{t_2}dt_1 f_2^x(t_2)f_1^x(t_1)t_1 = \int_0^Tdt_2\int_{0}^{t_2}dt_1 t_2^3 t_1^3 = T^7/28, \notag \\
    \int_0^Tdt_2\int_{0}^{t_2}dt_1 f_1^x(t_2)f_2^x(t_1)t_1 = \int_0^Tdt_2\int_{0}^{t_2}dt_1 t_2^2 t_1^4 = T^7/35.\label{eq:Example_Integration}
\end{align}
In this case, the contributions of the coefficients are distinct but linearly dependent. \par 

One can follow the following procedure to construct the terms in $H^{(l)}$ and its coefficients. We will exemplify the procedure on the Ising ansatz with global control fields: 
\begin{align} \label{eq:Ising_again}
    H_{\text{nat}} = J \sum_{i}Z_iZ_{i+1}\\
    H_{\text{ctr}}(t) = \sum_{i,\alpha}f_\alpha(t)S^\alpha_i.
\end{align}
\begin{figure}
    \centering
    \includegraphics[width=0.8\linewidth]{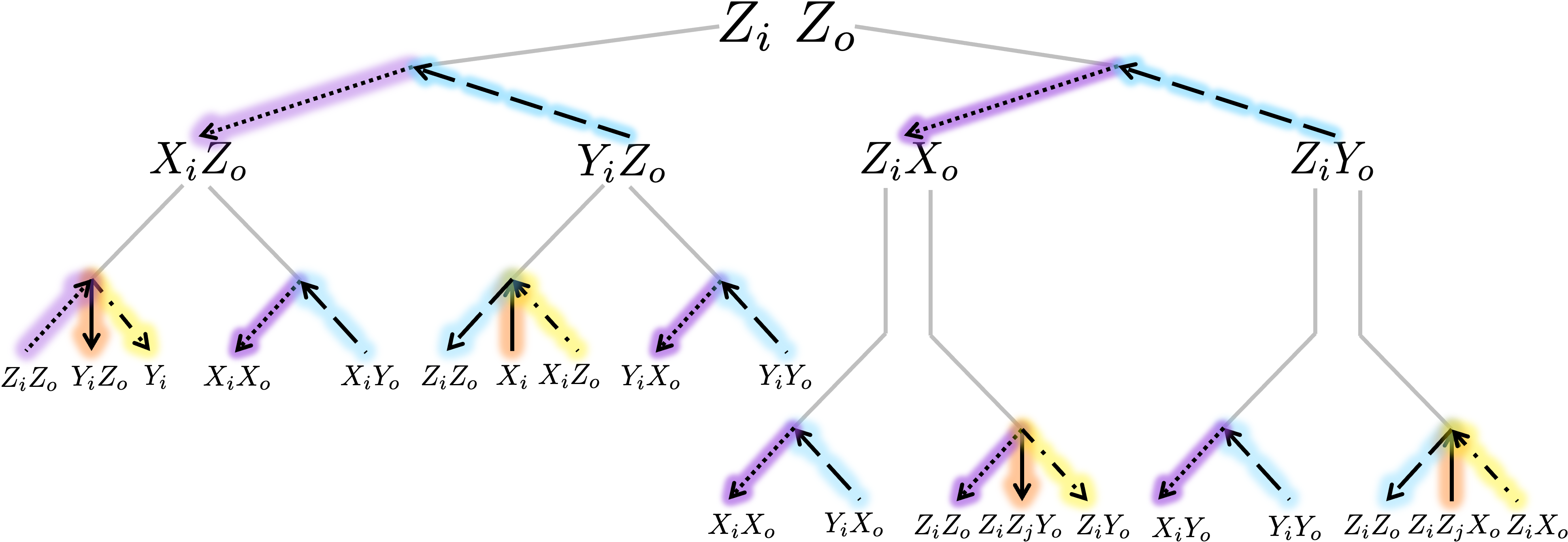}
    \caption{Diagramatic representation of the commutation relationships of the Hamiltonian composed of a nearest-neighbor Ising-type interaction and global controls along the $x,y$ and $z$ directions. Arrows represent directions where coefficients are multiplied by a $+1$ while directions opposite to the arrows acquire a factor of $-1$.}
    \label{fig:Commutators_Appdx}
\end{figure}

\begin{enumerate}
    \item First, we construct a graph that explains how the elements in $H_{\text{nat}}$ and $H_{\text{ctr}}$ transform under commutations. In our example, the operator $X_i$ transforms into $Y_i$ when commuted with $Z_i$, picking up a coefficient of $i f_z$. It also transforms to $Y_iZ_{i+1}$, picking up a coefficient of $iJ$. Fig.~\ref{fig:Commutators_Appdx} exemplifies the transformations generated by terms in our Hamiltonian. The arrows represent commutations that lead to the prescribed coefficient, while going in the opposite direction requires an extra factor of $-1$. 
    \item To construct a term given by $l$ commutation relations, we first grab an operator from the Hamiltonian (for example, $J Z_i Z_{i+1}$) and operate on it using the graph in Step 1 a total of $l$ times to produce one possible operator. We observe that many different operators can be created in this fashion, starting from a given initial operator. We must write the acquired coefficients from right to left at each stage. For time-dependent coefficients, the times acquire growing indices, from 0 to $l$, from right to left. For example, for $l=2$, we may obtain the following operators: 
    \begin{align}
        J Z_i Z_{i+1}&\longrightarrow -if_x(t_1)J Y_iZ_{i+1}\longrightarrow -f_x(t_2)f_x(t_1)J Y_i Y_{i+1} \\
        J Z_i Z_{i+1}&\longrightarrow i f_y(t_1)J Z_iX_{i+1}\longrightarrow -J f_y(t_1)J Z_iY_{i+1}Z_{i+2}.
    \end{align}
    \item We then multiply the coefficients by the constant $(-i)^l/(l!T)$ which weighs in the contribution to the Magnus expansion. Lastly, the coefficients are integrated and added to all other contributions paired with the same operator. For the paths exemplified above in Step 2, we see that $Z_iZ_{i+1}$ we get the coefficients
    \begin{align}
        \frac{1}{2T}\overline{f_xf_xJ} &= \frac{1}{2T}\int_0^T\int_{0}^{t_2}\int_{0}^{t_1}dt_2dt_1dt_0 f_x(t_2)f_x(t_1)J\\
        \frac{1}{2T}\overline{Jf_yJ} &= \frac{1}{2T}\int_0^T\int_{0}^{t_2}\int_{0}^{t_1}dt_2dt_1dt_0 Jf_y(t_1)J
    \end{align}
\end{enumerate}

\begin{table}[h!]
\begin{center}
\begin{tabular}{|c|c|}\hline
    Operator & Coefficient \\ \hline \hline
    $X_i$ & $\overline{f_x}+\overline{[f_z, f_x]}+\frac{1}{2}\overline{J[f_x,J]}+\frac{1}{2}\overline{f_y[f_x,f_y]}+\frac{1}{2}\overline{f_z[f_x,f_z]}$\\ \hline 
    $Y_i$ & $\overline{f_y} +\overline{[f_x, f_z]}+\frac{1}{2}\overline{J[f_y, J]}+\frac{1}{2}\overline{f_x[f_y, f_x]}+\frac{1}{2}\overline{f_z[f_y, f_z]}$ \\ \hline 
    $Z_i$ & $\overline{f_z}+\overline{[f_y, f_x]}+\frac{1}{2}\overline{f_x[f_z, f_x]}+\frac{1}{2}\overline{f_y[f_z,f_y]}$\\\hline
    $X_iX_{i+1}$ & $\frac{1}{2}\overline{f_zf_yJ} +\frac{1}{2}\overline{f_y[f_y,J]}$ \\ \hline 
    $X_iY_{i+1}$ & $-\frac{1}{2}\overline{f_yf_xJ}+\frac{1}{2}\overline{f_x[J,f_y]}$ \\ \hline 
    $X_i Z_{i+1}$ & $\overline{[J,f_y]} +\frac{1}{2}\overline{J[f_x, f_z]}+\frac{1}{2}\overline{f_z[f_x, J]}$ \\ \hline 
    $Y_iX_{i+1}$ & $-\frac{1}{2}\overline{f_xf_yJ}+\frac{1}{2}\overline{f_y[J,f_x]}$\\ \hline
    $Y_iY_{i+1}$ & $\frac{1}{2}\overline{f_xf_xJ} + \frac{1}{2}\overline{f_x[f_x,J]}$ \\ \hline 
    $Y_iZ_{i+1}$ & $\overline{[f_x,J]}+\frac{1}{2}\overline{J[f_y, f_z]}+\frac{1}{2}\overline{f_z[J, f_y]}$ \\ \hline 
    $Z_i X_{i+1}$ & $\frac{1}{2}\overline{f_zf_xJ}-\overline{f_y J}$ \\ \hline
    $Z_iY_{i+1}$  & $\frac{1}{2}\overline{f_zf_yJ}+\overline{f_xJ}$ \\ \hline 
    $Z_i Z_{i+1}$ & $\overline{J} +\frac{1}{2}\overline{f_x[J,f_x]}+\frac{1}{2}\overline{f_y[J,f_y]}+\frac{1}{2}\overline{(f_yf_y-f_xf_x)J}$ \\\hline
    $Z_iY_{i+1}Z_{i+2}$ & $ \frac{1}{2}\overline{Jf_yJ}$\\ \hline 
    $Z_iX_{i+1}Z_{i+2}$ & $ \frac{1}{2}\overline{Jf_xJ}$ \\\hline
\end{tabular}
\caption{Operators and coefficients of the Ising ansatz with global fields appearing on the Magnus expansion up to second order.}
\label{table:IsingGlobal}
\end{center}
\end{table}

\begin{figure}
    \centering
    \includegraphics[width=0.75\linewidth]{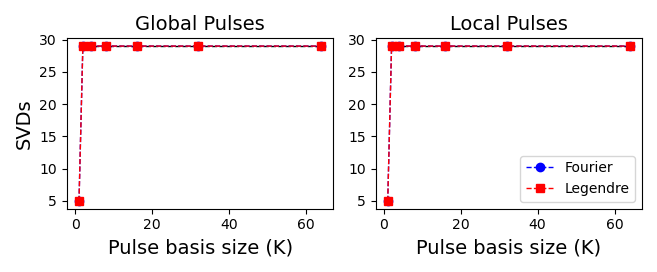}
    \caption{SVD analysis of the coefficients of the operators of the Magnus expansion up to second order. We chose 500 random instances of variational parameters and calculated all 29 coefficients of an Ising ansatz with global and local pulses for $N=3$. We plot the singular values of the matrix obtained by this procedure for different numbers of basis functions $K$ and for different functions (Fourier and Legendre). We find that whenever $K>1$, the coefficients are linearly independent. Similar results were observed for $N=4,5,6$.}
    \label{fig:SVDs}
\end{figure}

With this in mind, here are the coefficients of all the operators of the Magnus expansion up to $l=2$ for the model in Eq.~\ref{eq:Ising_again} as shown in Table \ref{table:IsingGlobal}, where we have used the notation $\overline{[A,B]} = \overline{AB}-\overline{BA}$. Notice that if the order of integration didn't matter, then $X_iZ_{i+1}$ and $Y_iZ_{i+1}$ would have a zero coefficient, and the coefficients of $X_iX_{i+1}$ and $X_iY_{i+1}$ would be the same (up to a sign). Therefore, if the order of integration didn't matter, the Magnus expansion would lead to linearly dependent coefficients, so individual terms couldn't be tuned independently. \par 
However, even when the order of integration matters, the analysis in Table~\ref{table:IsingGlobal} does not suffice to show that the coefficients are linearly independent. To do so, we see the right-most column of Table~\ref{table:IsingGlobal} as a map $\phi$ that takes in variational parameters $\bs{\theta}$ and produces the coefficients of the operators. In a sense, $\phi$ is a featured vector whose entries $\phi_{O}$ are the coefficients of the operator $O$. We can produce a data matrix $D$ whose $M$ columns correspond of choosing $M$ random instantiating of the variational parameters $\left\{\bs{\theta}_{i}\right\}_{i=1}^M$ and calculating the associated coefficient vectors $\left\{\phi(\bs{\theta}_i)\right\}_{i=1}^M$. We can then analyze the singular value decomposition of $D$ and count the number of nonzero singular values. Suppose the dimension of $\phi$ is $d$ and the number of nonzero singular values is $s$; if $s<d$, then the coefficients are linearly dependent. \par 
Fig.~\ref{fig:SVDs} shows the SVD analysis of the coefficients of the operators of the Magnus expansion up to the second order. We chose 500 random instances of variational parameters and calculated all 29 coefficients of an Ising ansatz with global and local pulses for $N=3$. We plot the singular values of the matrix obtained by this procedure for different numbers of basis functions $K$ and for different functions (Fourier and Legendre). We find that whenever $K>1$, the coefficients are linearly independent. Similar results were observed for $N=4,5,6$.\par 
\subsection{Spin-Squeezing using QPs}\label{A:Squeezing}
Suppose that all control fields are zero in a QP except $f_N^x$. Then, up to second order, the effective Hamiltonian has the following contributions: 
\begin{align*}
    H^{(0)} &= J \sum_i Z_iZ_N +F_N^x/TX_N \\
    H^{(1)} &= 0 \\ 
    H^{(2)} &= \frac{J}{2T}\sum_{i=1}^{N-1}\int_0^Tdt_2\int_0^{t_2}dt_1\int_0^{t_1}dt_0\left(f_N^x(t_2)f_N^x(t_1) - f_N^x(t_2)f_N^x(t_0)\right)Z_iZ_N \\
    &+\alpha_{ZZX}\sum_{ij}Z_iZ_jX_N.
\end{align*}
Thus, we have the effective Hamiltonian is 
\begin{equation}
    H_{\text{eff}} = \alpha_{ZZ}\sum_i Z_iZ_N + F_N^x/TX_N+\alpha_{ZZX}\sum_{ij}Z_iZ_jX_N.
\end{equation}
Let us analyze the case when $\alpha_{ZZ}=0$. Notice that 
\begin{equation*}
    \left(S_{\text{in}}^z\right)^2 = \left(\sum_{i=1}^{N-1} Z_i\right)\left(\sum_{j=1}^{N-1}  Z_j\right) = \sum_{ij}Z_iZ_j.
\end{equation*}
Thus, the Hamiltonian is that of Eq.~\ref{eq:Squeezing}. Notice that this hinges on $\alpha_{ZZ}=0$, which is thus a condition on the control fields we choose. 

\subsection{Additional Data on Simulating Jordan-Wigner Products}

\begin{figure}
    \centering
    \includegraphics[width=0.75\linewidth]{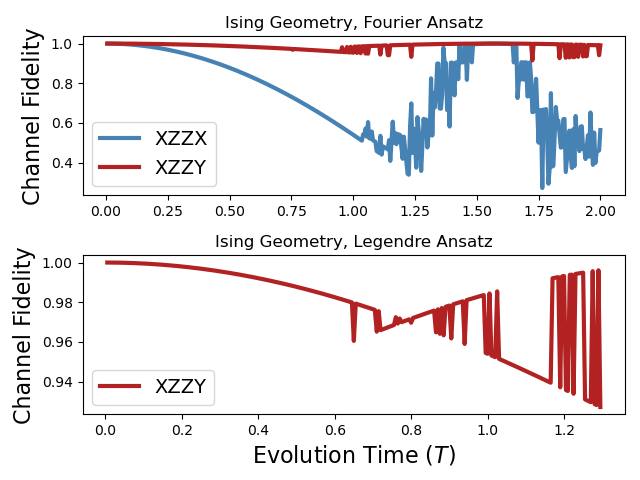}
    \caption{Results of simulating the evolution of Jordan-Wigner products for longer times and using the Fourier (top panel) and Legendre (bottom panel) ansatze}
    \label{fig:JWProds_Extended}
\end{figure}

In the main text, we showed the result of simulating the evolution of Jordan-Wigner products for small times and for the Fourier ansatz. Fig.~\ref{fig:JWProds_Extended} shows the results of simulating the evolution of Jordan-Wigner products for longer times and using the Fourier (top panel) and Legendre (bottom panel) ansatz. We observe that for both of these ansatz, the fidelity decreases as time increases but returns to maximum later. The time of this comeback differs from one to another anzats. The Fourier ansatz has a repetition property after $T=1$, but the Lengendre ansatz doesn't.

\end{document}